\documentclass[british]{article}
\usepackage{ae,aecompl}
\usepackage[T1]{fontenc}
\usepackage[latin9]{inputenc}
\usepackage{amsmath}
\usepackage{amsthm}
\usepackage{amssymb}
\usepackage{graphicx}
\usepackage{geometry}
\makeatletter

\providecommand{\tabularnewline}{\\}

\numberwithin{equation}{section}
\theoremstyle{plain}
\newtheorem{thm}{\protect\theoremname}
\theoremstyle{plain}
\newtheorem{lem}[thm]{\protect\lemmaname}
\theoremstyle{plain}
\newtheorem{prop}[thm]{\protect\propositionname}

\@ifundefined{date}{}{\date{}}

\usepackage{bbm}
\newtheorem{cond}[thm]{Condition}

\AddToHook{package/after/biblatex}{%
\DeclareNameAlias{sortname}{family-given}%
\DeclareNameAlias{author}{last-first}
}

\usepackage{babel}
\providecommand{\lemmaname}{Lemma}
\providecommand{\propositionname}{Proposition}
\providecommand{\theoremname}{Theorem}
\makeatother

\usepackage{babel}
\providecommand{\lemmaname}{Lemma}
\providecommand{\propositionname}{Proposition}
\providecommand{\theoremname}{Theorem}

\begin{document}
\global\long\def\IN{\mathbb{N}}%
\global\long\def\II{\mathbbm{1}}%
\global\long\def\IZ{\mathbb{Z}}%
\global\long\def\IQ{\mathbb{Q}}%
\global\long\def\IR{\mathbb{R}}%
\global\long\def\IC{\mathbb{C}}%
\global\long\def\IP{\mathbb{P}}%
\global\long\def\IE{\mathbb{E}}%
\global\long\def\IV{\mathbb{V}}%

\title{Models of Opinion Dynamics with Random Parametrisation}
\author{Gabor Toth\thanks{Corresponding author, email: gabor.toth@iimas.unam.mx}\\
{\footnotesize IIMAS-UNAM,}\\
{\footnotesize Escolar 3000, C.U., Coyoac�n, 04510, Mexico City, Mexico}}
\maketitle
\begin{abstract}
\noindent We analyse a generalisation of the Galam model of binary
opinion dynamics in which iterative discussions take place in local
groups of individuals and study the effects of random deviations from
the group majority. The probability of a deviation or flip depends
on the magnitude of the majority. Depending on the values of the flip
parameters which give the probability of a deviation, the model shows
a wide variety of behaviour. We are interested in the characteristics
of the model when the flip parameters are themselves randomly selected,
following some probability distribution. Examples of these characteristics
are whether large majorities and ties are attractors or repulsors,
or the number of fixed points in the dynamics of the model. Which
of the features of the model are likely to appear? Which ones are
unlikely because they only present as events of low probability with
respect to the distribution of the flip parameters? Answers to such
questions allow us to distinguish mathematical properties which are
stable under a variety of assumptions on the distribution of the flip
parameters from features which are very rare and thus more of theoretical
than practical interest. In this article, we present both exact numerical
results for specific distributions of the flip parameters and small
discussion groups and rigorous results in the form of limit theorems
for large discussion groups. Small discussion groups model friend
or work groups -- people that personally know each other and frequently
spend time together. Large groups represent scenarios such as social
media or political entities such as cities, states, or countries. 
\end{abstract}


\section{Introduction}

Sociophysics is the study of social phenomena by means of methods
developed for the study of the physical world. Among the topics of
interest in sociophysics lies the study of opinion dynamics (see e.g.
\cite{ChChCh2006,CFL2009,Schw2019}). One of the earliest contributions
to the field was the Galam model of opinion dynamics first introduced
in \cite{Gal1986}. Aside from the Galam model, the situations in
which opinion dynamics have been analysed by sociophysics is highly
varied: both the time component and the opinion space, i.e. the space
of all possible opinions an individual can hold, can be discrete or
continuous. For continuous opinion spaces, see the Friedkin-Johnsen
model \cite{FJ1990} and the bounded confidence models \cite{DNAW2000}.
Beyond sociophysics, opinion dynamics has been studied by scientists
from other fields, too: both psychologists and mathematicians have
made contributions. Some early contributions of social psychologists
are \cite{As1956,Ab1964,DGr1974}. Recent articles by mathematicians
studying opinion dynamics models are \cite{LT2015,AMBG2016,BBP2017},
which provide mathematically rigorous results about the models involved.

The Galam model features a binary choice facing each member of a large
population. Discussions occur in small groups of size $r\in\IN$ which
model groups of friends or work colleagues. Once the discussion finishes,
all members of the group adopt the majority opinion. Then the population
forms new discussion groups, and a new round of discussions takes
place with the adoption of the `local' i.e. discussion group majority
at the end of the round. We denote by the number $p_{0}\in\left[0,1\right]$
the proportion among the population preferring one of the alternatives,
say alternative $A$, before any discussions take place, and by $p_{t}$,
$t\in\IN$, the same proportion after discussion number $t$ and the
adoption of the local majority has taken place. Then the dynamics
of the Galam model can be analysed by determining the long term behaviour
of the proportion $p_{t}$, more precisely whether $\lim_{t\rightarrow\infty}p_{t}$
exists and if it does what its value is. The same mathematical formalism
can be used to study the different but related scenario of binary
voting in a hierarchy of individuals comprising several levels. Each
level features groups of $r$ individuals who elect a representative
for the next higher level who votes for the majority opinion in the
group. Instead of repeated discussions taking place with the population
subdivided into several small groups, here we determine which alternative
wins, i.e. the winning opinion at the top of the hierarchy. So instead
of the time dimension, we have a certain number of levels $n\in\IN$
in the hierarchy. Fixing $r$ and $n$ implies the number of voters
in the population is $r^{n}$. Similarly to the dynamics of $p_{t}$
in the opinion model, here we have $p_{k}$ for the proportion of
voters on level $k\in\left\{ 1,\ldots,n\right\} $ of the hierarchy
in favour of $A$, and we are interested in the behaviour of $p_{n}$.

The main distinction concerning the long term dynamics is whether
polarisation takes place, i.e. whether there is a tendency to a tie,
leaving the society in a state of conflict, or whether there is a
tendency towards a macroscopic majority in favour of one of the alternatives.
These long term dynamics are easy to determine for the Galam model:
suppose there is an initial majority in favour of one of the alternatives.
Then the long term tendency is towards a unanimous majority in favour
of the same alternative. Convergence is faster, the larger the size
$r$ of the discussion groups is.

To make the model produce qualitatively different behaviour, a number
of generalisations have been proposed. Using different versions of
this model, attempts have been made to predict the outcome of elections
and referenda in the works \cite{Lehir,Gal2017b,Gal2018}. One of
the generalisations of the basic model is the introduction of a tendency
to contrarianism, which manifests as the possibility of rejecting
the local majority in favour of the contrary alternative. These models
were first studied in \cite{Gal2002,Gal2004}. Contrarian tendencies
have been studied extensively by other authors using different models,
see e.g.\! \cite{TM2013,BR2015,GC2017,LHL2017,Ham2018}.

In this tradition of generalisations of the Galam model of opinion
dynamics, the so called local flip model with flips against the majority
where the likelihood of the flip depends on the magnitude of the majority
was introduced in \cite{TG2022}. This model allows us to study scenarios
where individuals are guided by the opinions of the others in their
local groups -- albeit not necessarily in a positive sense. It could
be the case that small majorities are not particularly persuasive:
individuals may not accept such close outcomes as sufficiently authoritative.
Or it could be that particularly large, close to unanimous, majorities
are distasteful to some individuals and induce resistance to the majority
decision. These two scenarios were called the `vertical' and `horizontal
frame', respectively, in \cite{TG2022}. These are aspects which could
not be explored using previous models of binary opinion dynamics,
and they shed some new light on several social phenomena triggered
by one or a few individuals acting against larger local majorities.

Let $r\in\IN$ be the size of the local discussion groups. Given a
current proportion $p\in\left[0,1\right]$ of the overall population
favouring alternative $A$, the proportion in favour after the next
round of discussions is given by the update function. The general
update function of the local flip model is 
\begin{equation}
R_{r,\mathbf{a}}\left(p\right)=\sum_{i=\frac{r+1}{2}}^{r}\left(\begin{array}{c}
r\\
i
\end{array}\right)\left[\left(1-a_{i}\right)p^{i}\left(1-p\right)^{r-i}+a_{i}p^{r-i}\left(1-p\right)^{i}\right].\label{eq:flip_voting}
\end{equation}
The vector $\mathbf{a}=\left(a_{\frac{r+1}{2}},\ldots,a_{r}\right)$
contains all the flip parameters of the model. The parameter $a_{i}$
gives the probability that the majority is not adopted given that
the majority is of size $i\in\left\{ \frac{r+1}{2},\ldots,r\right\} $.
When $i$ voters are in favour of $A$, and we encode each of these
votes as a $+1$ and each vote for $B$ as a $-1$, the sum of all
votes, called the voting margin, is $S=i-\left(r-i\right)=2i-r$.
The model is symmetric in that a flip against the local majority is
equally likely no matter the alternative the majority favours. Thus,
the flip parameters are chosen in such a way that they only depend
on the absolute voting margin $\left|S\right|=\left|2i-r\right|$
but not the sign of $S$. This symmetry implies that the point $p=1/2$,
which signifies a tie between the two alternatives, is a universal
fixed point in the dynamics of the model. The local flip model is
the most general model of binary opinion dynamics with random deviations
from the majority that is symmetric with respect to the two alternatives.

To make the local flip model more accessible than the mere statement
of the general update function accomplishes, we describe the simplest
case where the local discussion groups are of size 3. Then the model
has only two flip parameters: $a_{2}$ is the probability that a flip
against the group majority takes place conditional on there being
a 2-1 majority, and $a_{3}$ is the probability of a flip if there
is a unanimous majority. Table \ref{tab:Flip-Voting-r=00003D00003D3}
describes the process by which the opinions of the members of a local
discussion group are updated after discussion ends. The information
is read as follows: in the first two rows of the table, we have a
unanimous configuration of opinions in favour of $A$ in the local
discussion group. The probability that a randomly selected local discussion
group of size 3 has the configuration $AAA$, given that there is
a proportion $p$ of the overall population in favour of $A$, is
$p^{3}$. In the basic model, the majority would be adopted. However,
in the local flip model, there is a probability that this unanimous
majority is not adopted given by $a_{3}\in\left[0,1\right]$. Thus,
we obtain a probability of a group vote in favour of $A$ given by
$\left(1-a_{3}\right)p^{3}$ and a probability of going for $B$ given
by the complementary $a_{3}p^{3}$. The updated overall proportion
of individuals in favour of $A$, $R_{3,\mathbf{a}}\left(p\right)$,
is obtained by summing the probabilities in the rows where $A$ is
adopted. 
\begin{table}
\centering{}
\global\long\def\arraystretch{1.8}%
\begin{tabular}{|c|c|c|}
\hline 
Configuration & Group vote & Probability\tabularnewline
\hline 
\hline 
$AAA$ & $A$ & $\left(1-a_{3}\right)p^{3}$\tabularnewline
\hline 
 & $B$ & $a_{3}p^{3}$\tabularnewline
\hline 
$AAB\,\cdot3$ & $A$ & $\left(1-a_{2}\right)\cdot3p^{2}\left(1-p\right)$\tabularnewline
\hline 
 & $B$ & $a_{2}\cdot3p^{2}\left(1-p\right)$\tabularnewline
\hline 
$ABB\,\cdot3$ & $A$ & $a_{2}\cdot3p\left(1-p\right)^{2}$\tabularnewline
\hline 
 & $B$ & $\left(1-a_{2}\right)\cdot3p\left(1-p\right)^{2}$\tabularnewline
\hline 
$BBB$ & $A$ & $a_{3}\left(1-p\right)^{3}$\tabularnewline
\hline 
 & $B$ & $\left(1-a_{3}\right)\left(1-p\right)^{3}$\tabularnewline
\hline 
\end{tabular}\caption{\label{tab:Flip-Voting-r=00003D00003D3}Local Flip Model, $r=3$}
\end{table}

The aforementioned basic model and the contrarian model introduced
in \cite{Gal2002,Gal2004} are both special cases of the local flip
model. The basic model corresponds to the assumption $\mathbf{a}=0$,
meaning there are no flips against the majority under any circumstances.
The contrarian model features flat flip probabilities $\mathbf{a}=\left(a,\ldots,a\right)$
for some $a\in\left[0,1\right]$. Both of these models, being special
cases of the local flip model, share the universal fixed point $p=1/2$.
The basic model also has the fixed points 0 and 1. These features
lead to the convergence to unanimous majorities in the basic model
described above.

Instead of assuming some fixed vector of flip parameters $\mathbf{a}$
as in \cite{TG2022} and analysing the properties of the model under
that assumption, we investigate in this article how likely some of
the features of the model are when the parameters are randomly selected.
Therefore, in this article, $\mathbf{a}$ is a random vector that
follows some probability distribution. As a consequence, features
of the model, such as the stability parameter 
\begin{equation}
\lambda_{r}\left(\mathbf{a}\right):=R'_{r,\mathbf{a}}\left(1/2\right)=\frac{1}{2^{r-1}}\sum_{i=\frac{r+1}{2}}^{r}\left(\begin{array}{c}
r\\
i
\end{array}\right)\left(2i-r\right)\left(1-2a_{i}\right)\label{eq:Lambda_r}
\end{equation}
of the universal fixed point 1/2, will also be random variables. Other
features, such as the presence of unanimous attractors will be events
of the probability space on which $\mathbf{a}$ is defined. Hence,
it makes sense to ask ourselves the question of how likely these events
are, or how likely it is that the universal fixed point 1/2 is stable.

There are two possible interpretations of random flip parameters: 
\begin{enumerate}
\item Different issues to be discussed induce tendencies to adopt the majority
opinion or to resist it. In some cases, there may be a strong tendency
to align with the majority, and the flip parameters will be close
to 0. Other issues may trigger strong resistance against the majority
opinion, e.g.\! if the members of the majority are perceived to be
very arrogant or intolerant of divergent opinions. We consider that
for each possible issue there is some realisation of the flip parameters
$\mathbf{a}$ that affects the dynamics of the model. Thus, if we
study the local flip model with random flip parameters, we gain a
big picture understanding of how discussions take place and opinions
are shaped averaged over a large number of possible discussion topics. 
\item We can also consider the mathematical question of what the `typical
model' looks like. We know that the model exhibits a variety of features,
such as different numbers of fixed points with differing stability
properties, depending on the flip parameters, but how likely are these
different properties to occur? It may turn out that some features
-- although possible for special values of the flip parameters --
are exceedingly unlikely to occur under random flip parameters. In
that case, these features are more of theoretical than practical interest.
To make this distinction is especially important given the great generality
of the local flip model, which leads to an `anything goes' situation. 
\end{enumerate}
This article is organised in four sections and an appendix. After
this introduction, Section \ref{sec:Small} analyses features of the
local flip model with random flip parameters when the local discussion
groups are small. This setting allows us to assume a specific distribution
of the flip parameters and calculate the probability that the model
has certain characteristics explicitly. Afterwards, in Section \ref{sec:Large},
we analyse properties of the local flip model with random flip parameters
when the discussion groups are large. Section \ref{sec:Conclusion}
presents the conclusions we have reached. Finally, the Appendix contains
the proofs of the theorems presented in the article.

\section{\label{sec:Small}Small Local Discussion Groups}

In this section, we analyse properties of the versions of the local
flip model discussed in \cite{TG2022} under the assumption that the
flip parameters of the model are independent and uniformly distributed
on the interval $\left[0,1\right]$. We will denote the uniform distribution
as $\mathcal{U}$ followed by the support of the distribution, e.g.\!
$\left[0,1\right]$. Some versions of the model treated in \cite{TG2022}
involve restricted parameters, i.e.\! some of the flip parameters
are set equal to 0. The corresponding assumption under random flip
parameters is a Dirac distribution\footnote{For any $x\in\IR$, the Dirac distribution or point mass at $x$ is
defined as $\delta_{x}A:=\begin{cases}
1, & \textup{if }x\in A\\
0, & \textup{if }x\notin A
\end{cases}$ for all subsets $A$ of $\IR$.} at 0 (written as $\delta_{0}$) of those parameters which are 0.

The local discussion groups in this section are of size 3 or 5. This
is applicable to settings such as friend groups or colleagues at work.
In this setting, no matter how large the overall population is, the
discussion groups remain small and it is the number of discussion
groups that increases with the overall population. In Section \ref{sec:Large},
we will analyse the case of large discussion groups.

\subsection{Group Size $r=3$}

We first look at the model with local discussion groups of size $r=3$.
See Table \ref{tab:Flip-Voting-r=00003D00003D3} and Section 4.1 of
\cite{TG2022} for an analysis of this model. This model features
two flip parameters: we have $a_{2}$ which is the flip parameter
when there is a 2-1 majority in favour of one of the alternatives;
equivalently, the absolute voting margin $\left|S\right|$ is equal
to 1. The other flip parameter, $a_{3}$, describes the probability
of disregarding a unanimous 3-0 majority. We will use the notation
$a_{2}=a$ and $a_{3}=b$ as in \cite{TG2022} and assume a uniform
distribution, i.e.\! the random vector $\left(a,b\right)$ is uniformly
distributed on the set $\left[0,1\right]^{2}$, which we write as
$\left(a,b\right)\sim\mathcal{U}\left[0,1\right]^{2}$.

In order to determine the behaviour of the model for different values
of the flip parameters $a$ and $b$, the stability of the universal
fixed point 1/2 is calculated. There are four different regions in
the parameter space of this model $\left[0,1\right]^{2}$ with differing
stability properties: in the region $\mathbf{L}$ given by the inequality
$b<1/3-a$, 1/2 is unstable and the dynamics are monotonically away
from 1/2. The stable region $\mathbf{M}_{1}$ with monotonic dynamics
is given by the inequalities $1/3-a<b<1-a$. The other stable region
$\mathbf{M}_{2}$ with alternating dynamics is delimited by the inequalities
$1-a<b<5/3-a$. Finally, there is an unstable region $\mathbf{H}$
with alternating dynamics which lies above $b=5/3-a$. A graphical
representation of these results can be found in Figure \ref{fig:r=00003D00003D3_regimes}.

Having identified these regions, we can calculate the probability
that the fixed point 1/2 is stable or unstable: $\IP\left\{ 1/2\text{ is stable}\right\} =8/9$,
$\IP\left\{ 1/2\text{ is unstable}\right\} =1/9$. Under the uniform
distribution $\mathcal{U}\left[0,1\right]^{2}$, these probabilities
are simply the areas of the respective regions in the parameter space
$\left[0,1\right]^{2}$.

The other aspect analysed in \cite{TG2022} was the number of fixed
points. There are three regions: the region in which every value $p\in\left[0,1\right]$
is a fixed point consists of a single point, $\mathbf{F}_{\infty}:=\left\{ \left(1/3,0\right)\right\} $.
The region $\mathbf{F}_{1}$ with only a single fixed point which
is $1/2$ is given by the inequalities $1/3-a\leq b$ and $b>0$.
The region with three different fixed points is the complement $\mathbf{F}_{3}:=\left[0,1\right]^{2}\backslash\left(\mathbf{F}_{\infty}\cup\mathbf{F}_{1}\right)$,
i.e.\! the corner region around the origin and the $a$-axis excluding
the point: $\left(1/3,0\right)$.

Thus, the probabilities of the number of fixed points taking the values
$1,3,$ and $\infty$ are, respectively,\\
 $\IP\left\{ \textup{there is 1 fixed point}\right\} =17/18$, $\IP\left\{ \textup{there are 3 fixed points}\right\} =1/18$,
and\\
 $\IP\left\{ \textup{there are infinitely many fixed points}\right\} =0$.

As shown in the appendix A.2 of \cite{TG2022}, for any group size
$r$ the unanimous majority points 0 and 1 are fixed points if and
only if $a_{r}=0$. Furthermore, in case 0 and 1 are fixed points,
they are stable if and only if $a_{r-1}\leq1/r$. The event of having
unanimous attractors for group size $r=3$, i.e.\! the points 0 and
1 being fixed points and attractive, has probability 0: $\IP\left\{ \text{unanimous attractors}\right\} \leq\IP\left\{ a_{r}=0\right\} =0$.
\begin{figure}
\centering{}\includegraphics{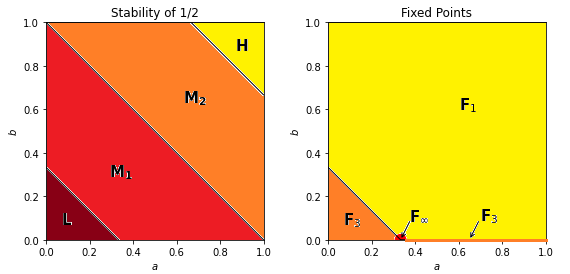}\caption{\label{fig:r=00003D00003D3_regimes}Regimes of the Local Flip Model,
$r=3$}
\end{figure}

\subsection{Group Size $r=5$}

We consider the same models as in Section 4.2 of \cite{TG2022} with
the same notation: $a_{3}=a$, $a_{4}=b$, and $a_{5}=c$. Each of
these models has some constraint on the possible values of the flip
parameters. Similarly to the discussion group size $r=3$, we can
conduct an analysis of the model's behaviour. See Section 4.2 of \cite{TG2022}
for the precise results, which, due to length constraints on this
article, we choose to omit. We present a graphical representation
of the version of the model with $c=0$, i.e.\! there are no flips
against unanimous majorities, in Figure \ref{fig:r=00003D00003D5_regimes}.

\begin{figure}
\centering{}\includegraphics{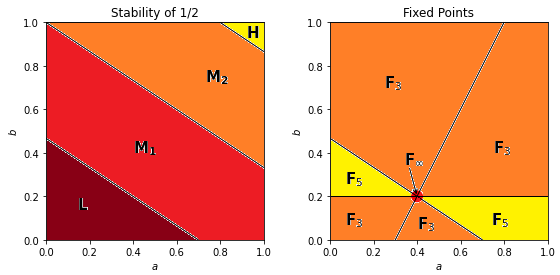}\caption{\label{fig:r=00003D00003D5_regimes}Regimes of the Local Flip Model,
$r=5$, $c=0$}
\end{figure}

The probabilities of a number of properties of four versions of local
flip models analysed in Section 4.2 of \cite{TG2022} can be found
in Table \ref{tab:Models-r=00003D00003D5}. These are exact figures
with no rounding applied. Due to the uniform distribution we have
assumed for the flip parameters, the probabilities presented in Table
\ref{tab:Models-r=00003D00003D5} are the areas of the corresponding
regions in the parameter space of the respective version of the model.
E.g., for the model with $\left(a,b\right)\sim\mathcal{U}\left[0,1\right]^{2},c=0$
presented in the last column in Table \ref{tab:Models-r=00003D00003D5},
the probabilities can be calculated directly from the diagrams in
Figure \ref{fig:r=00003D00003D5_regimes}.

\begin{table}
\begin{centering}
\global\long\def\arraystretch{1.5}%
\begin{tabular}{|c|c|c|c|c|}
\hline 
 & $a\sim\mathcal{U}\left[0,1\right],$  & $b\sim\mathcal{U}\left[0,1\right],$  & $c\sim\mathcal{U}\left[0,1\right],$  & $\left(a,b\right)\sim\mathcal{U}\left[0,1\right]^{2},$\tabularnewline
 & $b=c=0$  & $a=c=0$  & $a=b=0$  & $c=0$\tabularnewline
\hline 
\hline 
$\IP\left\{ 1/2\text{ is stable}\right\} $  & 3/10  & 8/15  & 0  & 247/300\tabularnewline
\hline 
$\IP\left\{ 1/2\text{ is unstable}\right\} $  & 7/10  & 7/15  & 1  & 53/300\tabularnewline
\hline 
$\IP\left\{ \textup{3 fixed points}\right\} $  & 7/10  & 11/15  & 1  & 1/12\tabularnewline
\hline 
$\IP\left\{ \textup{5 fixed points}\right\} $  & 3/10  & 4/15  & 0  & 11/12\tabularnewline
\hline 
$\IP\left\{ \text{unanimous attractors}\right\} $  & 1  & 1/5  & 0  & 17/100\tabularnewline
\hline 
\end{tabular}
\par\end{centering}
\caption{\label{tab:Models-r=00003D00003D5}Models with Discussion Groups of
Size $r=5$}
\end{table}

As we see, under uniform distribution of the flip parameters, there
generally is a small number of fixed points. Specifically, the probability
that a full complement of $r$ fixed points exists is small for both
$r=3$ and $r=5$, with the one exception of the distribution $\left(a,b\right)\sim\mathcal{U}\left[0,1\right]^{2},c=0$.
This exception is reached by setting the flip parameter for unanimous
majorities to the only value that allows for unanimous attractors.
This is thus not a typical result. In line with the generally small
number of fixed points, the unanimous majorities are rarely attractive
fixed points. This is thus a stable result even under randomly selected
flip parameters which clearly distinguishes the local flip model from
the basic Galam model of opinion dynamics, which features unanimous
attractors for any local discussion group size.

As for the universal fixed point 1/2, we observe mostly stable fixed
points. However, there are some exceptions when the right set of constraints
is imposed on the flip parameters. For $r=3$ and unrestricted flip
parameters, we have a high probability of a stable fixed point 1/2.
Similarly, for $r=5$ and $\left(a,b\right)\sim\mathcal{U}\left[0,1\right]^{2},c=0$,
we have a stable fixed point. On the other hand, if we set both $b$
and $c$ equal to 0, i.e.\! we only allow flips when the majority
is paper thin, an unstable fixed point 1/2 is more likely than a stable
one.

We could try using a different distribution for the flip parameters
to illustrate that these conclusions depend on the particular distribution
chosen. Due to the somewhat arbitrary nature of such an undertaking,
we choose instead to investigate the case of large update group sizes,
which allows more robust conclusions which do not require assuming
specific distributions to obtain results.

\section{\label{sec:Large}Large Local Discussion Groups}

In this section, we present results concerning mainly large local
discussion groups, which is understood to be the limit as $r$ goes
to infinity, although some of the following results hold for all values
of $r$. Contrary to the last section, which was about small discussion
groups, such as friend or work groups which are typically in the single
digits, here we are concerned with structures such as social media,
where large groups of individuals discuss topics -- frequently in
a very controversial manner. Another scenario would be identifying
the discussion groups with the states of a federal republic. As the
overall population becomes very large, we have either a bounded number
of local discussion groups or else the number of discussion groups
also grows without bound but more slowly than the overall population.

\subsection{Unanimous Attractors}

For any $r$, 0 and 1 are fixed points if and only if $a_{r}=0$ (see
appendix A2 of \cite{TG2022}). So $\IP\left\{ 0,1\textup{ fixed}\right\} =\IP\left\{ a_{r}=0\right\} $
depends only on the marginal distribution of the single parameter
$a_{r}$. How likely is it that 0 and 1 are fixed points and they
are attractors? This probability is given by 
\[
\IP\left\{ 0,1\textup{ stable and }0,1\textup{ fixed}\right\} =\IP\left\{ a_{r-1}\leq1/r\textup{ and }a_{r}=0\right\} .
\]

We consider two cases which represent the extremes among the possible
distributions of $\mathbf{a}$: 
\begin{enumerate}
\item $a_{i}$ are stochastically independent. Then 
\[
\IP\left\{ 0,1\textup{ stable and }0,1\textup{ fixed}\right\} =\IP\left\{ a_{r-1}\leq1/r\right\} \IP\left\{ a_{r}=0\right\} .
\]
\item All $a_{i}$ are equal almost surely. This is the case of the contrarian
model with a flat flip probability $a$. Then 
\[
\IP\left\{ 0,1\textup{ stable and }0,1\textup{ fixed}\right\} =\IP\left\{ a=0\right\} .
\]
\end{enumerate}
We conclude that unanimous attractors, which is a feature of the basic
model of opinion dynamics with no local flips, is relatively unlikely
to occur in both of these extreme scenarios. To find a stochastic
model for the flip parameters which makes unanimous attractors more
likely, we would have to look for some joint distribution that lies
somewhere in between the two extremes, tailoring it specifically to
achieve the desired outcome.

\subsection{Stability of the Fixed Point $p=1/2$}

\begin{figure}
\centering{}\includegraphics[scale=0.6]{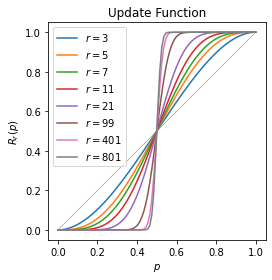}\caption{\label{fig:Update-Function}Update Function of the Basic Model}
\end{figure}

In Figure \ref{fig:Update-Function}, we see the update function (\ref{eq:flip_voting})
of the basic model, i.e.\! when $\mathbf{a}=0$. The stability parameter
$\lambda_{r}\left(0\right)$ of the universal fixed point defined
in (\ref{eq:Lambda_r}) is the slope of each curve at $p=1/2$. The
diagram illustrates that as the size of the local discussion group
increases, so does $\lambda_{r}\left(0\right)$. It was proved in
\cite{Gal2017} (see p. 27), that as $r$ goes to infinity, $\lambda_{r}\left(0\right)$
behaves asymptotically like the expression $\sqrt{\frac{2r}{\pi}}$.
So the fixed point 1/2 becomes more and more unstable, with the dynamics
being monotonic for any value of $r$, meaning that if the initial
proportion $p_{0}$ of the overall population in favour of $A$ starts
close by 1/2 -- but not exactly at a tie -- then for large local
discussion groups, the dynamics tend away from a tie to a unanimous
majority in favour of the initially favoured alternative, and convergence
is faster the larger $r$ is.

We start our investigation of the stability of the fixed point 1/2
by considering the two cases treated in the section about unanimous
attractors: independent flip parameters and flat flip parameters.
Afterwards, we investigate what happens in between these two extremes
using a model of ferromagnetism that features different degrees of
correlation between binary random variables.

\subsubsection{\label{subsec:Independent}Independent $a_{i}$}

Let $\left(a_{i}\right)_{i\in\IN}$ be an infinite sequence of stochastically
independent and identically distributed random variables. To avoid
dealing with the trivial case that there is no randomness in the model,
which is the case if the variance $\IV a_{i}$ of each flip parameter
is 0, we exclude this possibility by assuming $\IV a_{i}>0$ for the
rest of the article.

The key property of the flip parameters is their expectation $\IE a_{i}$.
As the flip parameters are bounded $\left[0,1\right]$-valued random
variables, this expectation always exists. If for each $i$ the expectation
is different than 1/2, then the sign of $\IE\lambda_{r}\left(\mathbf{a}\right)$
is given by the sign of $1-2\IE a_{i}$. However, this in of itself
is of limited usefulness. The entire distribution of the stability
parameter matters. Let us start with $\IV\lambda_{r}\left(\mathbf{a}\right)$,
the variance of $\lambda_{r}\left(\mathbf{a}\right)$. Its typical
magnitude will turn out to be smaller than $r$. This implies that
the stability parameter under the present assumptions scales more
slowly as $r$ goes to infinity compared to the deterministic basic
model. We have the following limit theorem from which the previous
claim follows. In the statement of the limit theorems in this article,
the symbol $\xrightarrow[r\rightarrow\infty]{\textup{d}}$ stands
for convergence in distribution as $r$ goes to infinity. For the
remainder of this article, we will be using the usual notation for
asymptotic expressions\footnote{Let $\left(a_{n}\right)_{n\in\IN}$ and $\left(b_{n}\right)_{n\in\IN}$
be sequences in the real numbers. Then we will write 
\begin{align*}
a_{n} & =o\left(b_{n}\right)\quad\textup{if}\quad\lim_{n\rightarrow\infty}\frac{a_{n}}{b_{n}}=0,\\
a_{n} & =\Theta\left(b_{n}\right)\quad\textup{if}\quad\liminf_{n\rightarrow\infty}\left|\frac{a_{n}}{b_{n}}\right|>0\quad\textup{and}\quad\limsup_{n\rightarrow\infty}\left|\frac{a_{n}}{b_{n}}\right|<\infty,\\
a_{n} & \approx b_{n}\quad\textup{if}\quad\lim_{n\rightarrow\infty}\frac{a_{n}}{b_{n}}=1.
\end{align*}
Informally, $a_{n}=o\left(b_{n}\right)$ means $a_{n}$ is `asymptotically
smaller' than $b_{n}$, and $a_{n}=\Theta\left(b_{n}\right)$ means
$a_{n}$ and $b_{n}$ are `asymptotically of the same order'. $a_{n}\approx b_{n}$
is a stronger condition than being of the same order. It means we
can asymptotically substitute $b_{n}$ for $a_{n}$.}. 
\begin{thm}
\label{thm:iid_neq_1/2}Let $\left(a_{i}\right)_{i\in\IN}$ be a sequence
of independent and identically distributed random variables with support
in $\left[0,1\right]$ and $\IE a_{i}\neq1/2$ for each $i\in\IN$.
Then the sequence of random variables $\left(\lambda_{r}\left(\mathbf{a}\right)\right)$
defined in (\ref{eq:Lambda_r}) and normalised by $\sqrt{r}$ converges
in distribution. More specifically, 
\[
\frac{\lambda_{r}\left(\mathbf{a}\right)}{\sqrt{r}}\xrightarrow[r\rightarrow\infty]{\textup{d}}\delta_{\sqrt{\frac{2}{\pi}}\left(1-2\IE a_{1}\right)}
\]
and $\IV\lambda_{r}\left(\mathbf{a}\right)=\Theta\left(\sqrt{r}\right)$. 
\end{thm}

On the other hand, if the expectations $\IE a_{i}$ are equal to 1/2,
then $\lambda_{r}\left(\mathbf{a}\right)$ is a centred random variable.
The fixed point 1/2 is therefore `on average stable'. However, this
does not imply that a stable fixed point is typical behaviour for
this model under these assumptions. A better description of the typical
behaviour is given by its variance, which is of order $\Theta\left(\sqrt{r}\right)$.
The formal result is the following limit theorem. Below and in the
rest of this article, $\mathcal{N}\left(\mu,\sigma^{2}\right)$ stands
for the normal distribution with expectation $\mu\in\IR$ and variance
$\sigma^{2}\geq0$. The special case $\sigma=0$ can be identified
with $\delta_{\mu}$. 
\begin{thm}
\label{thm:iid_eq_1/2}Let $\left(a_{i}\right)_{i\in\IN}$ be a sequence
of independent and identically distributed random variables with support
in $\left[0,1\right]$ and $\IE a_{i}=1/2$ for each $i\in\IN$. We
define the constants 
\[
b_{ri}:=\left(\begin{array}{c}
r\\
i
\end{array}\right)\frac{2i-r}{2^{r-1}}
\]
for all $r$ and all $i\in\left\{ \frac{r+1}{2},\ldots,r\right\} $
and the sequence $s_{r}:=\sqrt{\sum_{i=\frac{r+1}{2}}^{r}b_{ri}^{2}}$.
Asymptotically, $s_{r}$ is of order $\Theta\left(r^{1/4}\right)$.
The sequence of random variables $\left(\lambda_{r}\left(\mathbf{a}\right)\right)$
defined in (\ref{eq:Lambda_r}) and normalised by $s_{r}$ converges
in distribution to a centred normal distribution. Specifically, 
\[
\frac{\lambda_{r}\left(\mathbf{a}\right)}{s_{r}}\xrightarrow[r\rightarrow\infty]{\textup{d}}\mathcal{N}\left(0,4\IV a_{1}\right).
\]
Also, $\IV\lambda_{r}\left(\mathbf{a}\right)=\Theta\left(\sqrt{r}\right)$. 
\end{thm}

We can summarise our findings regarding independent flip parameters
as follows: 
\begin{enumerate}
\item If for each $i$ $\IE a_{i}\neq1/2$, then we have a bias towards
either going along with the majority or rejecting it. The high flip
probability regime, given by $\IE a_{i}>1/2$, leads to a negative
stability parameter. This means there are alternating dynamics of
the majorities around the fixed point 1/2, repeatedly switching between
a majority in favour of alternative $A$ and a majority in favour
of $B$. If $\IE a_{i}<1/2$ for each $i$, then the dynamics around
the fixed point 1/2 are monotonic. The typical magnitude of the stability
parameter scales with $r$ and is of order $\Theta\left(\sqrt{r}\right)$.
This is the same magnitude as in the basic model with all flip parameters
equal to 0, but keep in mind that in the basic model, the sign of
the stability parameter is always positive, whereas here the sign
is determined by $\IE a_{i}$.\\
 Informally, we can say that the typical realisation of the stability
parameter can be expressed as $\lambda_{r}\left(\mathbf{a}\right)=\sqrt{\frac{2r}{\pi}}\left(1-2\IE a_{1}\right)\pm\Theta\left(r^{1/4}\right)$. 
\item If $\IE a_{i}=1/2$ for each $i$, the stability parameter is a centred
random variable with equal probabilities of being positive or negative.
Its fluctuations are described by a normal distribution whose variance
is proportional to the variance of each flip parameter, and their
magnitude is of order $\Theta\left(\sqrt{r}\right)$.\\
 Informally, the typical magnitude of the stability parameter is given
by the expression $\lambda_{r}\left(\mathbf{a}\right)=\pm\Theta\left(r^{1/4}\right)$. 
\end{enumerate}
In both cases, the fixed point 1/2 becomes unstable as the discussion
groups grow large, meaning there is a tendency towards majorities
in favour of one of the alternatives, albeit the dynamics can be alternating
which implies there is no convergence of the distribution of opinions
towards a limit. Therefore, independent and identically distributed
flip parameters do not give rise to a tendency towards ties in the
dynamics of the model. Under independence of the flip parameters,
we would have to relax the identical distribution assumption and tailor
the distributions of the $a_{i}$ to obtain a stable fixed point 1/2
and thus a tendency towards ties.

\subsubsection{\label{subsec:Flat-Flip-Probabilities}Flat Flip Probabilities}

The other extreme is when there is perfect positive correlation between
all flip parameters as in the contrarian model. Let $a$ be the random
variable which gives the value for each flip parameter. Then, as $r$
grows large, the probability that we observe a model with a stable
universal fixed point 1/2 is given by $\IP\left\{ a=1/2\right\} $,
as only a realisation of $a=1/2$ will make 1/2 stable. All other
values of the flip parameters will lead to a highly unstable fixed
point, either with alternating dynamics if $a>1/2$ or monotonic dynamics
if $a<1/2$. So we can explicitly calculate the probabilities of certain
properties of the contrarian model, similarly to the case of small
$r$ treated above.

\subsubsection{Correlated Flip Parameters}

In between the two extremes of stochastic independence (i.e.\! all
flip parameters vary in value without any correlation to each other)
and the case of perfect correlation in which all flip parameters assume
the same value, we want to analyse a range of distributions for positively
correlated flip parameters. We cover the range of very weakly correlated
to strongly correlated flip parameters. The latter we understand as
the situation where the typical sum $\sum_{i=\frac{r+1}{2}}^{r}a_{i}$
is of order $r$, or `macroscopic', whereas in the former case that
same sum would be $o\left(r\right)$ with high probability.

To obtain varying degrees of positive correlation between the flip
parameters, we employ the Curie-Weiss model (CWM) of ferromagnetism,
as it allows us to observe three different regimes of correlations
while also being amenable to analytic solutions at least asymptotically.
The CWM describes a set of elementary magnets (or `spins') that tend
to align with each other. The two possible states for each magnet
are encoded in the random variables as the values $-1$ and $+1$.
The CWM was first applied to the study of problems in the social sciences
in \cite{BD2001}. Since then, there have been several articles using
the CWM to explore some social or economic phenomenon (see e.g.\!
\cite{Ki2007,CGh2007,OEA2018,LSV2020,To2020phd,KT2021c}).

Let $n\in\IN$ be the number of random variables $\left(X_{1}^{(n)},\ldots,X_{n}^{(n)}\right)$.
The CWM is defined by the so called `canonical ensemble', the probability
of each configuration $x\in\left\{ -1,1\right\} ^{n}$: 
\begin{equation}
\IP\left(X_{1}^{(n)}=x_{1},\ldots,X_{n}^{(n)}=x_{n}\right)=Z^{-1}\exp\left(\frac{\beta}{2n}\left(\sum_{i=1}^{n}x_{i}\right)^{2}\right),\label{eq:CWM}
\end{equation}
where $\beta\geq0$ is the inverse temperature parameter and $Z$
is a normalisation constant. $\beta$ regulates the interaction strength
between the random variables. For positive $\beta$, the configurations
with the highest probability are $\left(1,\ldots,1\right)$ and $\left(-1,\ldots,-1\right)$.
On a technical note, contrary to the case of independent random variables
where we can assume that there is a single infinite sequence of independent
random variables, of which we are free to take the first $n$ variables,
there is no such infinite sequence of Curie-Weiss random variables.
Instead, we have a separate model with $n$ variables for each $n\in\IN$.
In order to emphasize this difference, we included the superindex
$(n)$ in the notation.

The CWM has three regimes of distinct behaviour: 
\begin{enumerate}
\item For $\beta<1$, the interaction between the $X_{i}^{(n)}$ is weak. 
\item $\beta=1$ is a critical point with its own distinct behaviour. 
\item For $\beta>1$, the interaction is strong. 
\end{enumerate}
Each of these descriptions can be made precise in terms of the limiting
distribution of the so called magnetisation $\sum_{i=1}^{n}X_{i}^{(n)}/n$
as $n$ goes to infinity. We note that for $\beta=0$ the variables
$X_{i}^{(n)}$ are independent. Thus, we have a special case of our
assumptions in Section \ref{subsec:Independent}.

Now re-index the random variables $X_{i}^{(n)}$ such that the index
$i$ runs from $\frac{r+1}{2}$ to $r$ instead of $1$ to $n$. We
will write $X_{i}^{(r)}$ for each $r$ and each $i$. As $1-2a_{i}$
takes values in the interval $\left[-1,1\right]$, we can identify
$1-2a_{i}$ with $X_{i}^{(r)}$ and obtain a stochastic model of the
flip parameters with 
\[
\lambda_{r}\left(\mathbf{a}\right)=\frac{1}{2^{r-1}}\sum_{i=\frac{r+1}{2}}^{r}\left(\begin{array}{c}
r\\
i
\end{array}\right)\left(2i-r\right)X_{i}^{(r)}.
\]
However, due to the binary nature of the random variables $X_{i}^{(r)}$
we only get the two levels of flip probabilities $X_{i}^{(r)}=1=1-2a_{i}$,
which is equivalent to $a_{i}=0$, and $X_{i}^{(r)}=-1$, which is
equivalent to $a_{i}=1$. This all-or-nothing setup is somewhat extreme.
It only considers the possibilities of either respecting the majority
when $a_{i}=0$ or flipping with probability 1. Hence, we introduce
a scale $s\in\left(0,1\right)$, which is constant, and we identify
$1-2a_{i}$ with $sX_{i}$, thus obtaining two possible levels of
flip probabilities $\frac{1-s}{2}$ and $\frac{1+s}{2}$. So there
is the possibility of a low flip probability and a high flip probability,
both of which lie in $\left(0,1\right)$. The stability parameter
is therefore 
\[
\lambda_{r}\left(\mathbf{a}\right)=\frac{s}{2^{r-1}}\sum_{i=\frac{r+1}{2}}^{r}\left(\begin{array}{c}
r\\
i
\end{array}\right)\left(2i-r\right)X_{i}^{(r)}.
\]
Note that $\IE a_{i}=1/2$ holds for all $i$, and hence $\IE\lambda_{r}\left(\mathbf{a}\right)=0$.

The results for this setup are as follows: 
\begin{enumerate}
\item For $\beta=0$, the model behaves just as outlined in Section \ref{subsec:Independent}.
The stability parameter $\lambda_{r}\left(\mathbf{a}\right)$ normalised
by a term of order $r^{1/4}$ tends to a normal distribution.\\
 For $0<\beta<1$, there is positive correlation between the flip
parameters $a_{i}=\frac{1-sX_{i}^{(r)}}{2}$. However, this correlation
is relatively weak, and the stability parameter behaves as for independent
flip parameters. 
\item If $\beta=1$, then there is stronger correlation between the flip
parameters than for $\beta<1$. However, this increase in correlation
is not reflected in the typical magnitude of $\lambda_{r}\left(\mathbf{a}\right)$.
We once again find a typical magnitude of order $\Theta\left(r^{1/4}\right)$
with equal probabilities of $\lambda_{r}\left(\mathbf{a}\right)$
being positive or negative. It is only the multiplicative constant
in the term $\Theta\left(r^{1/4}\right)$ which differs. See Theorem
\ref{thm:CWM} for more details. 
\item For $\beta>1$, the correlation between the $a_{i}$ is stronger than
for $\beta=1$. Now we see this reflected in the typical magnitude
of $\lambda_{r}\left(\mathbf{a}\right)$ which is of order $\Theta\left(\sqrt{r}\right)$.
Thus, we conclude that the strong correlation between the flip parameters
induces the same order of instability of the fixed point as we observe
in the basic model. However, due to the expectation of each flip parameter
being 1/2, we once again have equal probabilities of observing monotonic
dynamics around 1/2 or alternating dynamics. 
\end{enumerate}
As we see, using the CWM for the flip parameters, we cover the cases
of $\lambda_{r}\left(\mathbf{a}\right)$ being typically of order
$\Theta\left(r^{1/4}\right)$ similarly to the case of independent
flip parameters and $\lambda_{r}\left(\mathbf{a}\right)$ being typically
of order $\Theta\left(\sqrt{r}\right)$ which is the same magnitude
as in the basic model. The stronger the correlation between the flip
parameters, the more unstable the fixed point 1/2 becomes. This means
for strongly correlated flip parameters, there is a pronounced tendency
to a large majority in favour of one of the alternatives, or else
there is oscillating behaviour that alternates between majorities
in favour of each of the alternatives in turn. For weak correlation,
signified by $\beta\leq1$, we obtain an unstable fixed point 1/2,
where the magnitude scales more slowly than for the strong correlation
regime $\beta>1$. Note that similarly to the case of independent
flip parameters with $\IE a_{i}=1/2$, the CWM also yields a random
sign of $\lambda_{r}\left(\mathbf{a}\right)$ with each sign having
probability 1/2.

Formally, we have the following limit theorem: 
\begin{thm}
\label{thm:CWM}Let for each $r$, $\left(X_{i}^{(r)}\right)_{\frac{r+1}{2}\leq i\leq r}$
be $\left\{ -1,1\right\} $-valued random variables with joint distribution
given by (\ref{eq:CWM}), $s\in\left(0,1\right)$, and let $a_{i}:=\frac{1-sX_{i}^{(r)}}{2}$
for each $i\in\left\{ \frac{r+1}{2},\ldots,r\right\} $. We define
the constants 
\[
b_{ri}:=s\left(\begin{array}{c}
r\\
i
\end{array}\right)\frac{2i-r}{2^{r-1}}
\]
for all $r$ and all $i\in\left\{ \frac{r+1}{2},\ldots,r\right\} $
and the sequence $s_{r}:=\sqrt{\sum_{i=\frac{r+1}{2}}^{r}b_{ri}^{2}}$.
Asymptotically, $s_{r}$ is of order $\Theta\left(r^{1/4}\right)$.
We have the following limiting distributions: 
\begin{enumerate}
\item If $\beta<1$, then 
\begin{equation}
\frac{\lambda_{r}\left(\mathbf{a}\right)}{s_{r}}\xrightarrow[r\rightarrow\infty]{\textup{d}}\mathcal{N}\left(0,1\right).\label{eq:CLT_lambda_r}
\end{equation}
\item If $\beta=1$, then $\frac{\lambda_{r}\left(\mathbf{a}\right)}{s_{r}}$
converges in distribution to a symmetric non-normal distribution. 
\item If $\beta>1$, then 
\[
\frac{\lambda_{r}\left(\mathbf{a}\right)}{\sqrt{r}}\xrightarrow[r\rightarrow\infty]{\textup{d}}\frac{1}{2}\left(\delta_{-m}+\delta_{m}\right),
\]
where $m>0$ is a constant independent of $r$ and strictly increasing
in $\beta$. 
\end{enumerate}
\end{thm}

Informally, if $\beta\leq1$, the typical stability parameter value
is $\lambda_{r}\left(\mathbf{a}\right)=\pm\Theta\left(r^{1/4}\right)$.
For $\beta>1$, the typical value is $\lambda_{r}\left(\mathbf{a}\right)=\pm m\sqrt{r}\pm\Theta\left(r^{1/4}\right)$.

The limiting distributions of (a suitably normalised) $\lambda_{r}\left(\mathbf{a}\right)$
broadly follow the same pattern as for the magnetisation 
\[
S_{r}:=\sum_{i=\frac{r+1}{2}}^{r}X_{i}^{(r)}.
\]
There are, however, some interesting differences: 
\begin{enumerate}
\item The limiting distribution of the normalised magnetisation $\sqrt{\frac{2}{r}}S_{r}$
when $\beta<1$ is also normal, but the variance depends on the parameter
$\beta$ (see Sections IV.4 and V.9 of \cite{Ell1985} for a detailed
analysis of the CWM): 
\begin{equation}
\sqrt{\frac{2}{r}}S_{r}\xrightarrow[r\rightarrow\infty]{\textup{d}}\mathcal{N}\left(0,\frac{1}{1-\beta}\right).\label{eq:CLT_magnetisation}
\end{equation}
The variance is 1 for $\beta=0$, i.e.\! when the $X_{i}^{(r)}$
are independent, and then it increases as $\beta$ grows, diverging
to infinity as $\beta$ approaches 1 from below. This is not the case
for $\frac{\lambda_{r}\left(\mathbf{a}\right)}{s_{r}}$, which converges
in distribution to a standard normal, independently of the value of
$\beta<1$. Note that $s_{r}$ does not depend on $\beta$. 
\item The normalisation required in each regime of the model in order to
obtain convergence in distribution is different. For the stability
parameter $\lambda_{r}\left(\mathbf{a}\right)$, we normalise by $s_{r}=\Theta\left(r^{1/4}\right)$
in both the regimes $\beta<1$ and $\beta=1$, while $\beta>1$ requires
normalisation by $\sqrt{r}$. The magnetisation $S_{r}$, on the other
hand, requires normalisation by $\sqrt{\frac{r}{2}}$, $\left(\frac{r}{2}\right)^{3/4}$,
and $\frac{r}{2}$, respectively, in each regime. This difference
is crucial since it directly affects the typical magnitude of the
stability parameter which is different than that of the magnetisation. 
\end{enumerate}
For independent flip parameters, we analysed the case that $\IE a_{i}=1/2$,
which makes $\lambda_{r}\left(\mathbf{a}\right)$ a centred random
variable, and the case $\IE a_{i}\neq1/2$. We can do the same for
flip parameters with a joint Curie-Weiss distribution. The first case
was treated above. If we allow an external magnetic field in the model,
the random variables $X_{i}^{(r)}$ no longer have expectation 0,
and thus the flip parameters $a_{i}$ will have expectation different
than 1/2. The joint distribution is given by 
\begin{equation}
\IP\left(X_{1}^{(n)}=x_{1},\ldots,X_{n}^{(n)}=x_{n}\right)=Z^{-1}\exp\left(\frac{\beta}{2n}\left(\sum_{i=1}^{n}x_{i}\right)^{2}+h\sum_{i=1}^{n}x_{i}\right),\label{eq:CWM_ext}
\end{equation}
for each configuration $x\in\left\{ -1,1\right\} ^{n}$. Above, the
parameter $h\in\IR$ gives the strength of the external magnetic field.
We once again re-index the variables such that for each $r$ the index
$i$ of $X_{i}^{(r)}$ runs from $\frac{r+1}{2}$ to $r$. In the
context of the flip parameters $a_{i}=\frac{1-sX_{i}^{(r)}}{2}$,
a positive $h$ implies a bias towards lower flip probabilities, whereas
a negative $h$ means there is a bias towards higher flip probabilities.
Hence, we can regard $h$ as a parameter that regulates how contrarian
the people are.

Under the presence of such a bias, the stability parameter will be
of order $\Theta\left(\sqrt{r}\right)$, regardless what the value
of the parameter $\beta\geq0$ is. The sign of $\lambda_{r}\left(\mathbf{a}\right)$
will be the same as that of $h$. Thus, for $h>0$ which is low contrarianism,
we obtain a model which is qualitatively similar to the basic model,
featuring an unstable fixed point 1/2 with monotonic dynamics. A value
$h<0$ which indicates high contrarianism, on the other hand, yields
an unstable fixed point with alternating dynamics. The magnitude of
the stability parameter will be $\Theta\left(\sqrt{r}\right)$ for
all $h\neq0$ and all $\beta\geq0$.

The formal result is the following theorem: 
\begin{thm}
\label{thm:CWM_ext}Let for each $r$, $\left(X_{i}^{(r)}\right)_{\frac{r+1}{2}\leq i\leq r}$
be $\left\{ -1,1\right\} $-valued random variables with joint distribution
given by (\ref{eq:CWM_ext}), $h\neq0$, $s\in\left(0,1\right)$,
and let $a_{i}:=\frac{1-sX_{i}^{(r)}}{2}$ for each $i\in\left\{ \frac{r+1}{2},\ldots,r\right\} $.
Let $b_{ri}$ and $s_{r}$ be defined as in Theorem \ref{thm:CWM}
for all $r$ and all $i\in\left\{ \frac{r+1}{2},\ldots,r\right\} $.
Then we have for all $\beta\geq0$ 
\begin{equation}
\frac{\lambda_{r}\left(\mathbf{a}\right)-x\left(\beta,h\right)\sum_{i=\frac{r+1}{2}}^{r}b_{ri}}{s_{r}}\xrightarrow[r\rightarrow\infty]{\textup{d}}\mathcal{N}\left(0,1\right).\label{eq:CLT_lambda_r_ext}
\end{equation}
Above, $x\left(\beta,h\right)$ is a constant independent of $r$
that has the same sign as $h$. $\left|x\left(\beta,h\right)\right|$
is increasing in $\beta$. 
\end{thm}

The proofs of this theorem and all previous ones can be found in the
Appendix.

Informally, for $h\neq0$, the typical value of the stability parameter
is $\lambda_{r}\left(\mathbf{a}\right)=x\left(\beta,h\right)\sqrt{\frac{2r}{\pi}}\pm\Theta\left(r^{1/4}\right)$.

Contrary to the absence of an external magnetic field, the CWM behaves
in a qualitatively similar fashion for all values of $\beta\geq0$
when there is an external magnetic field $h\neq0$. This is similar
to the behaviour of the magnetisation $S_{r}$. A central limit theorem
similar to the statement (\ref{eq:CLT_magnetisation}) holds for any
value $\beta\geq0$ provided that $h\neq0$. (\ref{eq:CLT_lambda_r_ext})
is a central limit theorem that holds for all values of $\beta$ and
$h\neq0$. However, just as without the external magnetic field, the
stability parameter and the magnetisation have different typical magnitudes:
the typical magnitude of $\lambda_{r}\left(\mathbf{a}\right)$ is
of order $\Theta\left(\sqrt{r}\right)$, whereas the typical magnitude
of $S_{r}$ is of order $\Theta\left(r\right)$ under the assumption
of an external magnetic field.

Even the strong correlation regime of the CWM does not give a stable
fixed point 1/2 with positive probability. One reason for this is
that there are no realisations of the vector of flip parameters $\mathbf{a}$
where each entry is equal to 1/2. Also, the typical realisations of
$\mathbf{a}$ feature a sizeable majority of entries $a_{i}$ pointing
in the same direction. In conclusion, strong correlation on its own
is not necessarily enough to reproduce the result we saw in Section
\ref{subsec:Flat-Flip-Probabilities} of a stable fixed point under
perfect positive correlation of the flip parameters \emph{and} a positive
probability of all of them being equal to 1/2.

\section{\label{sec:Conclusion}Conclusion}

We investigated the behaviour of the local flip model when the flip
parameters, which are ordinarily fixed constants, are randomly selected.
First we studied the case of small local discussion groups such as
real-life friend or work groups in Section \ref{sec:Small}. Under
independent uniformly distributed flip parameters, we found that 
\begin{itemize}
\item The number of fixed points in the dynamics of the model is generally
small (meaning smaller than the possible maximum of $r$ fixed points).
Thus, the dynamics of the model are fairly simple, with large basins
of attraction to the few attractors that mark the limits public discourse
tends towards as time passes. 
\item Unless we exclude the possibility of flips against unanimous majorities,
the local flip model typically does not present unanimous attractors.
This is a major difference compared to the basic model. As unanimous
majorities are extremely rare in real life, this is a point in favour
of the local flip model. 
\item The universal fixed point 1/2 is mostly stable. There is a tendency
towards ties in most discussions. 
\end{itemize}
Then, in Section \ref{sec:Large}, we studied large discussion groups.
Unanimous majorities are very unlikely to occur for both large and
small discussion groups. On the other hand, we found that the picture
differs considerably between small and large groups as far as the
stability of the fixed point 1/2 is concerned. Whereas the former
feature mostly stable fixed points at 1/2, the latter do not. In fact,
we found that the larger the local discussion groups, generally, the
more unstable the fixed point 1/2 becomes. The many degrees of freedom
inherent in large discussion groups do not lead to the coexistence
of stable and unstable fixed point at 1/2. Contrary to the basic model,
the sign of the stability parameter, which determines whether the
dynamics close to a tie are monotonic or alternating, can be positive
or negative. We showed that for unbiased flip parameters, i.e.\!
flip parameters with expectation 1/2, monotonic and alternating dynamics
around a tie are equally likely. If the flip parameters are biased,
then the sign of the bias determines the sign of the stability parameter
and thus whether we only see monotonic or alternating dynamics but
never both with positive probability. If the flip parameters are biased
towards accepting the majority, then the local flip model behaves
similarly to the basic model with monotonic dynamics. These findings
are stable across a range of degrees of dependency of the flip parameters,
all the way from independence to strong correlation in the low temperature
regime of the CWM. Thus, the picture the local flip model with random
flip parameters paints is that for small discussion groups there is
a tendency towards ties, whereas for large discussion groups there
is a tendency to move away from ties. We either see convergence to
a majority in favour of one of the alternatives, or we observe perpetual
swings from majorities in favour of one alternative to majorities
for the other one.

\section*{Acknowledgements}

This research was supported by a Secihti (formerly Conahcyt) postdoctoral
fellowship and through a SNII fellowship (candidate level) of the
author.

\appendix

\section*{Appendix}

We prove the results presented in Theorems \ref{thm:iid_neq_1/2},
\ref{thm:iid_eq_1/2}, \ref{thm:CWM}, and \ref{thm:CWM_ext}. We
will use the notation $a\wedge b=\min\{a,b\}$ and $a\vee b=\max\{a,b\}$
for all real numbers $a,b$ throughout the rest of this article.

\section{Proof of Theorem \ref{thm:iid_neq_1/2}}

The first two of these theorems pertain to sums of independent random
variables, and therefore we can apply the central limit theorem by
Feller-L�vy (see e.g.\! Theorem 6.16 in \cite{Kal2021}). The sequence
of random variables for which we want to determine a limiting distribution
is 
\begin{align*}
\frac{\lambda_{r}\left(\mathbf{a}\right)}{\sqrt{r}} & =\frac{1}{\sqrt{r}}\sum_{i=\frac{r+1}{2}}^{r}s\left(\begin{array}{c}
r\\
i
\end{array}\right)\frac{2i-r}{2^{r-1}}\left(1-2a_{i}\right)=\frac{1}{\sqrt{r}}\sum_{i=\frac{r+1}{2}}^{r}b_{ri}\left(1-2a_{i}\right),
\end{align*}
where the coefficients $b_{ri}$ equal $\left(\begin{array}{c}
r\\
i
\end{array}\right)\frac{2i-r}{2^{r-1}}$ as defined in the statement of Theorem \ref{thm:iid_eq_1/2}. We
define for each $r$ and each $i\in\left\{ \frac{r+1}{2},\ldots,r\right\} $
the random variable $Y_{ri}:=\frac{b_{ri}\left(1-2a_{i}\right)}{\sqrt{r}}$.
In order to show Theorem \ref{thm:iid_neq_1/2} for this triangular
array of random variables $\left(Y_{ri}\right)_{r\in\IN,i\in\left\{ \frac{r+1}{2},\ldots,r\right\} }$,
we need to verify four statements:

\begin{cond}\label{cond:Feller-L=00003D0000E9vy} \begin{enumerate} 

\item $\left(Y_{ri}\right)_{r\in\IN,i\in\left\{ \frac{r+1}{2},\ldots,r\right\} }$
is a null array, i.e.\! for each $r$ the random variables $Y_{ri}$
are independent and $\lim_{r\rightarrow\infty}\sup_{i}\IE\left(\left|Y_{ri}\right|\wedge1\right)=0$. 

\item $\sum_{i=\frac{r+1}{2}}^{r}\IP\left\{ \left|Y_{ri}\right|>\varepsilon\right\} \xrightarrow[r\rightarrow\infty]{}0,\varepsilon>0$. 

\item $\sum_{i=\frac{r+1}{2}}^{r}\IE\left(Y_{ri}\II_{\left\{ \left|Y_{ri}\right|\leq1\right\} }\right)\xrightarrow[r\rightarrow\infty]{}b\in\IR$,
where for any measurable set $A$ $\II_{A}$ refers to the indicator
function of $A$. 

\item $\sum_{i=\frac{r+1}{2}}^{r}\IV\left(Y_{ri}\II_{\left\{ \left|Y_{ri}\right|\leq1\right\} }\right)\xrightarrow[r\rightarrow\infty]{}c\geq0$.
\end{enumerate} \end{cond}

Once the four conditions are verified, we have obtained convergence
in distribution of $\sum_{i=\frac{r+1}{2}}^{r}Y_{ri}$ to $\mathcal{N}\left(b,c\right)$.
This includes the degenerate normal distribution $\mathcal{N}\left(b,0\right)=\delta_{b}$
as a special case.

In order to show the first statement, we need to determine the asymptotic
order of each $Y_{ri}$. That means we need to evaluate each $b_{ri}$
asymptotically. We will next present and prove several statements
regarding these coefficients before proceeding to the proof of the
four statements in Condition \ref{cond:Feller-L=00003D0000E9vy}.
Since we will later also need powers $b_{ri}^{k}$ for all $k\in\IN$,
we will take care of this in the following lemma: 
\begin{lem}
\label{lem:bri_asymp}Let $r,k\in\IN$ and $i=\frac{r+1}{2}+\eta_{r}\in\left\{ \frac{r+1}{2},\ldots,r\right\} $.
We distinguish four classes of $\eta_{r}$: 
\begin{enumerate}
\item If $\eta_{r}=o\left(\sqrt{r}\right)$, then $b_{ri}^{k}=\Theta\left(\frac{\eta_{r}^{k}}{r^{k/2}}\right)$. 
\item If $\lim_{r\rightarrow\infty}\frac{\eta_{r}}{\sqrt{r}}=h>0$, then
$b_{ri}^{k}\xrightarrow[r\rightarrow\infty]{}\left(\frac{32}{\pi}\right)^{k/2}\exp\left(-2kh^{2}\right)h^{k}>0$. 
\item If $\sqrt{r}=o\left(\eta_{r}\right)$ and $\eta_{r}=o(r)$, then $b_{ri}^{k}=\Theta\left(\exp\left(-2k\frac{\eta_{r}^{2}}{r}\right)\frac{\eta_{r}^{k}}{r^{k/2}}\right)$. 
\item If $\lim_{r\rightarrow\infty}\frac{\eta_{r}}{r}=h\in\left(1/2,1\right]$,
then $b_{ri}^{k}=\Theta\left(\exp\left(-2khr\right)r^{k/2}\right)$. 
\end{enumerate}
\end{lem}

\begin{proof}
The proof of this lemma uses the local limit theorem for the binomial
distribution by de Moivre-Laplace: 
\begin{thm}
\label{thm:Moivre-Laplace}Let $P_{n}$ be the binomial distribution
with $n\in\IN$ $\left\{ 0,1\right\} $-valued variables and parameter
$1/2$ and let $\phi$ be the Lebesgue density function of the standard
normal distribution. Then we have 
\[
\sup_{i\in\left\{ 0,\ldots,n\right\} }\left|\frac{\sqrt{n}}{2}P_{n}\left\{ i\right\} -\phi\left(\frac{i-n/2}{\sqrt{n}/2}\right)\right|\xrightarrow[r\rightarrow\infty]{}0.
\]
\end{thm}

The theorem states that the scaled point probabilities $\frac{\sqrt{n}}{2}P_{n}\left\{ i\right\} $
of the binomial distribution are well approximated for large values
of $n$ by the corresponding values of the density function of the
standard normal distribution. Note that the convergence is uniform
over all values of $i\in\left\{ 0,\ldots,n\right\} $.

An elementary calculation shows that for all $r$ and all $i$ 
\[
b_{ri}=4P_{r}\left\{ i\right\} \left(i-r/2\right).
\]
Translating Theorem \ref{thm:Moivre-Laplace} to our setting, we have
\begin{equation}
\sup_{i\in\left\{ \frac{r+1}{2},\ldots,r\right\} }\left|b_{ri}-4\phi\left(\frac{i-r/2}{\sqrt{r}/2}\right)\frac{i-r/2}{\sqrt{r}/2}\right|\xrightarrow[r\rightarrow\infty]{}0.\label{eq:bri_conv}
\end{equation}
Due to the exponential nature of $\phi$, the set 
\begin{equation}
A:=\left\{ \left.4\phi\left(\frac{i-r/2}{\sqrt{r}/2}\right)\frac{i-r/2}{\sqrt{r}/2}\,\right|\,r\in\IN,i\in\left\{ \frac{r+1}{2},\ldots,r\right\} \right\} \label{eq:values_normal_density}
\end{equation}
is bounded. Hence, the uniform convergence stated in (\ref{eq:bri_conv})
follows from Theorem \ref{thm:Moivre-Laplace}.

Now we evaluate the expression $4\phi\left(\frac{i-n/2}{\sqrt{n}/2}\right)\frac{i-n/2}{\sqrt{n}/2}$
to calculate the asymptotic expressions in the lemma. Let $r\in\IN$
and $i=\frac{r+1}{2}+\eta_{r}\in\left\{ \frac{r+1}{2},\ldots,r\right\} $.
We calculate 
\begin{align*}
b_{ri} & \approx4\phi\left(\frac{i-r/2}{\sqrt{r}/2}\right)\frac{i-r/2}{\sqrt{r}/2}=2\sqrt{\frac{2}{\pi}}\exp\left(-\left(\frac{\frac{r+1}{2}+\eta_{r}-\frac{r}{2}}{\sqrt{r}/2}\right)^{2}\right)\frac{\frac{r+1}{2}+\eta_{r}-\frac{r}{2}}{\sqrt{r}/2}\\
 & =2\sqrt{\frac{2}{\pi}}\exp\left(-\frac{1+4\eta_{r}+4\eta_{r}^{2}}{2r}\right)\frac{1+2\eta_{r}}{\sqrt{r}}.
\end{align*}
If $\eta_{r}=o\left(\sqrt{r}\right)$, then 
\[
b_{ri}\approx2\sqrt{\frac{2}{\pi}}\frac{2\eta_{r}}{\sqrt{r}}=\Theta\left(\frac{\eta_{r}}{\sqrt{r}}\right).
\]
If $\lim_{r\rightarrow\infty}\frac{\eta_{r}}{\sqrt{r}}=h>0$, then
\[
b_{ri}\approx2\sqrt{\frac{2}{\pi}}\exp\left(-\frac{4\eta_{r}^{2}}{2r}\right)\frac{2\eta_{r}}{\sqrt{r}}\xrightarrow[r\rightarrow\infty]{}\left(\frac{32}{\pi}\right)^{1/2}\exp\left(-2h^{2}\right)h,
\]
which is clearly positive. If $\sqrt{r}=o\left(\eta_{r}\right)$ and
$\eta_{r}=o(r)$, then 
\[
b_{ri}\approx2\sqrt{\frac{2}{\pi}}\exp\left(-\frac{4\eta_{r}^{2}}{2r}\right)\frac{2\eta_{r}}{\sqrt{r}}=\Theta\left(\exp\left(-2\frac{\eta_{r}^{2}}{r}\right)\frac{\eta_{r}}{\sqrt{r}}\right).
\]
If $\lim_{r\rightarrow\infty}\frac{\eta_{r}}{r}=h\in\left(1/2,1\right]$,
then 
\[
b_{ri}\approx2\sqrt{\frac{2}{\pi}}\exp\left(-2h-2hr\right)\cdot2h\sqrt{r}=\Theta\left(\exp\left(-2hr\right)\sqrt{r}\right).
\]
Thus, the statements of Lemma \ref{lem:bri_asymp} hold for $k=1$.
The statements for general $k\in\IN$ follow from the observation
that the set $A$ in (\ref{eq:values_normal_density}) is bounded,
and the function $x\mapsto x^{k}$ restricted to any bounded subset
of the reals is uniformly continuous. 
\end{proof}
Lemma \ref{lem:bri_asymp} states that the only coefficients $b_{ri}$
which do not decay as $r$ goes to infinity are those for which $i-\frac{r+1}{2}$
is of order $\sqrt{r}$, in which case $b_{ri}$ converges to some
positive constant. Therefore, if we are looking for the maximum of
$b_{ri}$ for a fixed $r$ over all $i$, we only need to look at
those $i$ which are of order $\sqrt{r}$, provided $r$ is large
enough. 
\begin{lem}
\label{lem:bri_max}Let $k\in\IN$. For large enough $r$, the maximum
of $\left\{ b_{ri}^{k}\,|\,i\in\left\{ \frac{r+1}{2},\ldots,r\right\} \right\} $
is located at $i\approx\frac{r+1}{2}+\frac{\sqrt{r}}{2}$ and the
value of the corresponding $b_{ri}^{k}$ is asymptotically given by
\[
0<\left(\frac{2\sqrt{2}}{\sqrt{e\pi}}\right)^{k}<1.
\]
Also, for all $r$ large enough, the set of coefficients $b_{ri}^{k}$
for all $i\in\left\{ \frac{r+1}{2},\ldots,r\right\} $ is bounded
above by some constant strictly smaller than $1$. 
\end{lem}

\begin{proof}
Let $k\in\IN$. Let $h>0$ and $i=\frac{r+1}{2}+\left\lfloor h\sqrt{r}\right\rfloor $.
The brackets $\left\lfloor \cdot\right\rfloor $ denote the floor
function which rounds down the expression on the inside to the nearest
integer. Then by Lemma \ref{lem:bri_asymp}, 
\[
b_{ri}^{k}\xrightarrow[r\rightarrow\infty]{}\left(\frac{32}{\pi}\right)^{k/2}\exp\left(-2kh^{2}\right)h^{k}.
\]
We identify the value of the constant $h$ which maximises the limit
on the right hand side above. This expression can be understood as
a function $f:\left(0,\infty\right)\rightarrow\IR$ of $h$. Calculating
the first derivative, we obtain 
\[
f'(h)=\exp\left(-2kh^{2}\right)k\left(-4h^{k+1}+h^{k-1}\right).
\]
Equating the first derivative to 0, we identify the critical point
$h_{0}=1/2$. We calculate the second derivative 
\[
f''(h)=\exp\left(-2kh^{2}\right)k\left(-4kh^{k}\left(-4h^{2}+1\right)+h^{k-2}\left(-4\left(k+1\right)h^{2}+k-1\right)\right).
\]
The value of the second derivative at the critical point $h_{0}$
is negative, hence $h_{0}$ is indeed the unique global maximum of
$f$ on $\left(0,\infty\right)$. By substituting $h_{0}$ into $f$,
we obtain the asymptotic maximum of $b_{ri}^{k}$. It is easily calculated
that 
\[
f\left(h_{0}\right)=\left(\frac{2\sqrt{2}}{\sqrt{e\pi}}\right)^{k},
\]
and this values lies in $\left(0,1\right)$. As for any $h>0$ and
$i=\frac{r+1}{2}+\left\lfloor h\sqrt{r}\right\rfloor $ the coefficients
converge uniformly over all values of $h$ to $f(h)$, we can choose
$r_{0}\in\IN$ large enough that for all $r\geq r_{0}$ and all $i\in\left\{ \frac{r+1}{2},\ldots,r\right\} $
we have 
\[
0<b_{ri}^{k}<1\quad\textup{and}\quad\sup_{r\geq r_{0}}\sup_{i\in\left\{ \frac{r+1}{2},\ldots,r\right\} }b_{ri}^{k}<1.
\]
So starting at $r_{0}$, all coefficients $b_{ri}^{k}$ are uniformly
bounded away from 1. 
\end{proof}
We next determine the asymptotic order of the sums $\sum_{i=\frac{r+1}{2}}^{r}b_{ri}^{k}$. 
\begin{lem}
\label{lem:sum_bri^k}Let $k\in\IN$. Then we have 
\[
\sum_{i=\frac{r+1}{2}}^{r}b_{ri}^{k}=\Theta\left(\sqrt{r}\right).
\]
\end{lem}

The lemma states that for all $k\in\IN$ the asymptotic order of the
sum $\sum_{i=\frac{r+1}{2}}^{r}b_{ri}^{k}$ is the same, namely $\Theta\left(\sqrt{r}\right)$.
This has important implications for the remaining proofs. 
\begin{proof}
We first show an asymptotic upper bound for $\sum_{i=\frac{r+1}{2}}^{r}b_{ri}^{k}$.
Lemma \ref{lem:bri_max} says that, for all $r$ large enough, all
summands $b_{ri}$ are smaller than 1. Hence, $b_{ri}^{k}\leq b_{ri}$
for all $i$. Therefore, we have the upper bound 
\[
\sum_{i=\frac{r+1}{2}}^{r}b_{ri}^{k}\leq\sum_{i=\frac{r+1}{2}}^{r}b_{ri}\approx\sqrt{\frac{2r}{\pi}}=\Theta\left(\sqrt{r}\right),
\]
using the asymptotic expression for $\sum_{i=\frac{r+1}{2}}^{r}b_{ri}$
from p.\! 27 of \cite{Gal2017}.

Next, we show a lower bound for $\sum_{i=\frac{r+1}{2}}^{r}b_{ri}^{k}$.
By Statement 2 of Lemma \ref{lem:bri_asymp}, we have for any constants
$0<c<C<\infty$, $r$ large enough, and $i\in B:=\left\{ i\in\IN\,|\,\frac{r+1}{2}+c\sqrt{r}<i<\frac{r+1}{2}+C\sqrt{r}\right\} $
\[
b_{ri}^{k}\geq\frac{1}{2}\min\left\{ \left.\left(\frac{32}{\pi}\right)^{k/2}\exp\left(-2kh^{2}\right)h^{k}\,\right|\,h\in\left[c,C\right]\right\} =:\tau.
\]
Then $\tau>0$ holds. We note that the cardinality of index set $B$
is at least $\left\lfloor \left(C-c\right)\sqrt{r}\right\rfloor $,
so a lower bound for the sum $\sum_{i=\frac{r+1}{2}}^{r}b_{ri}^{k}$
is given by 
\[
\tau\left\lfloor \left(C-c\right)\sqrt{r}\right\rfloor =\Theta\left(\sqrt{r}\right).
\]
\end{proof}

\subsubsection*{Statement 1}

After these preparatory lemmas which we will also use for later proofs,
we are ready to show Statement 1 concerning the $\left(Y_{ri}\right)_{r\in\IN,i\in\left\{ \frac{r+1}{2},\ldots,r\right\} }$
being a null array.

By Lemma \ref{lem:bri_max}, there is a constant $K>0$ such that
\begin{equation}
\left|Y_{ri}\right|=\frac{b_{ri}\left|1-2a_{i}\right|}{\sqrt{r}}\leq\frac{K}{\sqrt{r}}\label{eq:UB_Yri}
\end{equation}
holds for all $r$ and all $i\in\left\{ \frac{r+1}{2},\ldots,r\right\} $.
Therefore, we have 
\[
\IE\left(\left|Y_{ri}\right|\wedge1\right)=\IE\left(\frac{b_{ri}\left|1-2a_{i}\right|}{\sqrt{r}}\wedge1\right)\leq\frac{K}{\sqrt{r}}\xrightarrow[r\rightarrow\infty]{}0
\]
uniformly over all $r\in\IN,i\in\left\{ \frac{r+1}{2},\ldots,r\right\} $.
This proves Statement 1 in Condition \ref{cond:Feller-L=00003D0000E9vy}.

\subsubsection*{Statement 2}

Let $\varepsilon>0$. By (\ref{eq:UB_Yri}), we have for all $r>K^{2}/\varepsilon^{2}$
\[
\left|Y_{ri}\right|<\varepsilon
\]
for all $i\in\left\{ \frac{r+1}{2},\ldots,r\right\} $. Therefore,
\[
\sum_{i=\frac{r+1}{2}}^{r}\IP\left\{ \left|Y_{ri}\right|>\varepsilon\right\} =0
\]
for all $r>K^{2}/\varepsilon^{2}$. This proves Statement 2 in Condition
\ref{cond:Feller-L=00003D0000E9vy}.

\subsubsection*{Statement 3}

By (\ref{eq:UB_Yri}), for all $r\geq K^{2}$, 
\begin{align*}
\sum_{i=\frac{r+1}{2}}^{r}\IE\left(Y_{ri}\II_{\left\{ \left|Y_{ri}\right|\leq1\right\} }\right) & =\sum_{i=\frac{r+1}{2}}^{r}\IE\left(Y_{ri}\right)=\frac{1}{\sqrt{r}}\sum_{i=\frac{r+1}{2}}^{r}b_{ri}\IE\left(1-2a_{i}\right)\\
 & =\IE\left(1-2a_{1}\right)\frac{1}{\sqrt{r}}\sum_{i=\frac{r+1}{2}}^{r}b_{ri}\xrightarrow[r\rightarrow\infty]{}\sqrt{\frac{2}{\pi}}\left(1-2\IE a_{1}\right)=:b\in\IR.
\end{align*}

\subsubsection*{Statement 4}

We again employ (\ref{eq:UB_Yri}) to calculate that for all $r\geq K^{2}$
\begin{align}
\sum_{i=\frac{r+1}{2}}^{r}\IV\left(Y_{ri}\II_{\left\{ \left|Y_{ri}\right|\leq1\right\} }\right) & =\sum_{i=\frac{r+1}{2}}^{r}\IV Y_{ri}=\sum_{i=\frac{r+1}{2}}^{r}\left(\IE Y_{ri}^{2}-\left(\IE Y_{ri}\right)^{2}\right)\nonumber \\
 & \approx\frac{1}{r}\sum_{i=\frac{r+1}{2}}^{r}b_{ri}^{2}\left(\IE\left(1-2a_{i}\right)^{2}-\left(1-2\IE a_{i}\right)^{2}\right)\nonumber \\
 & =\frac{1}{r}\sum_{i=\frac{r+1}{2}}^{r}b_{ri}^{2}\left(4\IE a_{i}^{2}-4\left(\IE a_{i}\right)^{2}\right)\nonumber \\
 & =\frac{4\IV a_{1}}{r}\sum_{i=\frac{r+1}{2}}^{r}b_{ri}^{2}\xrightarrow[r\rightarrow\infty]{}0.\label{eq:var_iid}
\end{align}
In the last step, we used Lemma \ref{lem:sum_bri^k}, by which $\sum_{i=\frac{r+1}{2}}^{r}b_{ri}^{2}=\Theta\left(\sqrt{r}\right)$.

In virtue of the verification of Condition \ref{cond:Feller-L=00003D0000E9vy},
we have thus shown the limit theorem 
\[
\frac{\lambda_{r}\left(\mathbf{a}\right)}{\sqrt{r}}\xrightarrow[r\rightarrow\infty]{\textup{d}}\mathcal{N}\left(\sqrt{\frac{2}{\pi}}s\left(1-2\IE a_{1}\right),0\right)=\delta_{\sqrt{\frac{2}{\pi}}\left(1-2\IE a_{1}\right)}.
\]
The last statement in Theorem \ref{thm:iid_neq_1/2}, $\IV\lambda_{r}\left(\mathbf{a}\right)=\Theta\left(\sqrt{r}\right)$,
follows from our calculation (\ref{eq:var_iid}), Lemma \ref{lem:sum_bri^k},
and 
\[
\IV\lambda_{r}\left(\mathbf{a}\right)=r\sum_{i=\frac{r+1}{2}}^{r}\IV Y_{ri}=4\IV a_{1}\sum_{i=\frac{r+1}{2}}^{r}b_{ri}^{2}=\Theta\left(\sqrt{r}\right).
\]
This concludes the proof of Theorem \ref{thm:iid_neq_1/2}.

\section{Proof of Theorem \ref{thm:iid_eq_1/2}}

Using the lemmas from the last section, we now show Theorem \ref{thm:iid_eq_1/2}
by verifying Condition \ref{cond:Feller-L=00003D0000E9vy} for the
triangular array $\left(Y_{ri}\right)_{r\in\IN,i\in\left\{ \frac{r+1}{2},\ldots,r\right\} }$
defined by $Y_{ri}:=\frac{b_{ri}\left(1-2a_{i}\right)}{s_{r}}$ under
the assumption that for all $i$ $\IE a_{i}=1/2$ holds.

\subsubsection*{Statement 1}

We have to show that $\left(Y_{ri}\right)_{r\in\IN,i\in\left\{ \frac{r+1}{2},\ldots,r\right\} }$
is a null array. By Lemma \ref{lem:bri_max}, there is a constant
$K>0$ such that 
\begin{equation}
\left|Y_{ri}\right|=\frac{b_{ri}\left|1-2a_{i}\right|}{s_{r}}\leq\frac{K}{s_{r}}\label{eq:UB_Yri-1}
\end{equation}
holds for all $r$ and all $i=\frac{r+1}{2}$. We recall the definition
of $s_{r}:=\sqrt{\sum_{i=\frac{r+1}{2}}^{r}b_{ri}^{2}}$. Using Lemma
\ref{lem:sum_bri^k}, we obtain $s_{r}=\Theta\left(r^{1/4}\right)$,
and by Lemma \ref{lem:bri_max} we have 
\[
\IE\left(\left|Y_{ri}\right|\wedge1\right)=\IE\left(\frac{b_{ri}\left|1-2a_{i}\right|}{s_{r}}\wedge1\right)=\Theta\left(1/r^{1/4}\right)\xrightarrow[r\rightarrow\infty]{}0
\]
uniformly over all $r\in\IN,i\in\left\{ \frac{r+1}{2},\ldots,r\right\} $.
This proves Statement 1 in Condition \ref{cond:Feller-L=00003D0000E9vy}.

\subsubsection*{Statement 2}

Let $\varepsilon>0$. $s_{r}=\Theta\left(r^{1/4}\right)$ implies
$\text{\ensuremath{\tau}}:=\liminf_{r\rightarrow\infty}\frac{s_{r}}{r^{1/4}}>0$.
Hence there is an $r_{0}\in\IN$ such that for all $r\geq r_{0}$
the inequality $\tau/2<s_{r}/r^{1/4}$ holds. It follows from (\ref{eq:UB_Yri-1})
that, for all $r>r_{0}\vee\left(\frac{2K}{\tau\varepsilon}\right)^{4}$
and all $i\in\left\{ \frac{r+1}{2},\ldots,r\right\} $, we have 
\begin{equation}
\left|Y_{ri}\right|\leq\frac{K}{s_{r}}<\frac{2K}{\tau r^{1/4}}<\varepsilon.\label{eq:Yri_eps}
\end{equation}
As a consequence, 
\[
\sum_{i=\frac{r+1}{2}}^{r}\IP\left\{ \left|Y_{ri}\right|>\varepsilon\right\} =0
\]
for all $r>r_{0}\vee\left(\frac{2K}{\tau\varepsilon}\right)^{4}$.
This proves Statement 2 in Condition \ref{cond:Feller-L=00003D0000E9vy}.

\subsubsection*{Statement 3}

Using (\ref{eq:Yri_eps}) with $\varepsilon=1$, we have for all $r>r_{0}\vee\left(\frac{2K}{\tau}\right)^{4}$
\begin{align*}
\sum_{i=\frac{r+1}{2}}^{r}\IE\left(Y_{ri}\II_{\left\{ \left|Y_{ri}\right|\leq1\right\} }\right) & =\sum_{i=\frac{r+1}{2}}^{r}\IE\left(Y_{ri}\right)=\frac{1}{s_{r}}\sum_{i=\frac{r+1}{2}}^{r}b_{ri}\IE\left(1-2a_{i}\right)\\
 & =\IE\left(1-2a_{1}\right)\frac{1}{s_{r}}\sum_{i=\frac{r+1}{2}}^{r}b_{ri}=0=:b\in\IR.
\end{align*}

\subsubsection*{Statement 4}

We again employ (\ref{eq:Yri_eps}) to calculate that for all $r>r_{0}\vee\left(\frac{2K}{\tau}\right)^{4}$
\begin{align}
\sum_{i=\frac{r+1}{2}}^{r}\IV\left(Y_{ri}\II_{\left\{ \left|Y_{ri}\right|\leq1\right\} }\right) & =\sum_{i=\frac{r+1}{2}}^{r}\IV Y_{ri}=\sum_{i=\frac{r+1}{2}}^{r}\left(\IE Y_{ri}^{2}-\left(\IE Y_{ri}\right)^{2}\right)\nonumber \\
 & \approx\frac{1}{s_{r}^{2}}\sum_{i=\frac{r+1}{2}}^{r}b_{ri}^{2}\left(\IE\left(1-2a_{i}\right)^{2}-\left(1-2\IE a_{i}\right)^{2}\right)\nonumber \\
 & =\frac{1}{s_{r}^{2}}\sum_{i=\frac{r+1}{2}}^{r}b_{ri}^{2}\left(4\IE a_{i}^{2}-4\left(\IE a_{i}\right)^{2}\right)\nonumber \\
 & =\frac{4\IV a_{1}}{s_{r}^{2}}\sum_{i=\frac{r+1}{2}}^{r}b_{ri}^{2}=4\IV a_{1}.\label{eq:var_iid-1}
\end{align}
In the last step, we used the definition of $s_{r}$.

We have verified Condition \ref{cond:Feller-L=00003D0000E9vy} and
thus shown the limit theorem 
\[
\frac{\lambda_{r}\left(\mathbf{a}\right)}{s_{r}}\xrightarrow[r\rightarrow\infty]{\textup{d}}\mathcal{N}\left(0,4\IV a_{1}\right).
\]
The last statement in Theorem \ref{thm:iid_eq_1/2}, $\IV\lambda_{r}\left(\mathbf{a}\right)=\Theta\left(\sqrt{r}\right)$,
follows from our calculation (\ref{eq:var_iid-1}), Lemma \ref{lem:sum_bri^k},
and 
\[
\IV\lambda_{r}\left(\mathbf{a}\right)=s_{r}^{2}\sum_{i=\frac{r+1}{2}}^{r}\IV Y_{ri}=4\IV a_{1}\sum_{i=\frac{r+1}{2}}^{r}b_{ri}^{2}=\Theta\left(\sqrt{r}\right).
\]
This concludes the proof of Theorem \ref{thm:iid_eq_1/2}.

\section{Proof of Theorem \ref{thm:CWM}}

Theorems \ref{thm:CWM} and \ref{thm:CWM_ext} are limit theorems
for sums of triangular arrays of dependent random variables. Therefore,
we cannot employ the classic central limit theorem by Feller-L�vy.
Instead, we will use the method of moments. Let $\left(Y_{ni}\right)_{n\in\IN,i\in\left\{ 1,\ldots,n\right\} }$
be a triangular array of random variables. The method of moments consists
of showing the convergence of moments 
\[
\IE\left(\sum_{i=1}^{n}Y_{ni}\right)^{k}\xrightarrow[n\rightarrow\infty]{}m_{k}\in\IR
\]
for each $k\in\IN$. Provided the constants $m_{k}$ satisfy appropriate
upper bounds, e.g. 
\begin{equation}
m_{k}\leq AC^{k}k!,\quad k\in\IN,\label{eq:moderate_growth}
\end{equation}
for fixed constants $A,C\in\IR$, the convergence of the moments implies
convergence in distribution of the sequence $\left(\sum_{i=1}^{n}Y_{ni}\right)_{n\in\IN}$
to a limiting distribution $\mu$ with moments of all orders $k$
given by the limits $m_{k}$ above. Moreover, the limiting distribution
$\mu$ is uniquely determined by its moments $m_{k}$. Consult e.g.\!
Example 4 on p.\! 205 of \cite{RS1975} for a proof of these statements.
We note that all limiting distributions demonstrated in this article
satisfy the condition (\ref{eq:moderate_growth}), so the method of
moments is applicable.

Let $\beta\leq1$. We define $\left(Y_{ri}\right)_{r\in\IN,i\in\left\{ \frac{r+1}{2},\ldots,r\right\} }$
by $Y_{ri}:=\frac{b_{ri}\left(1-2a_{i}\right)}{s_{r}}=\frac{b_{ri}X_{ri}}{s_{r}}$
for all $r\in\IN$ and all $i\in\left\{ \frac{r+1}{2},\ldots,r\right\} $.
To calculate the limits of the moments $\IE\left(\sum_{i=\frac{r+1}{2}}^{r}Y_{ri}\right)^{k}$,
we have to evaluate asymptotically sums of the type 
\[
\sum_{i_{1},\ldots,i_{k}=\frac{r+1}{2}}^{r}\IE Y_{ri_{1}}\cdots Y_{ri_{k}}.
\]
Therefore, we will need to know the asymptotic behaviour of correlations
such as 
\[
\IE Y_{ri_{1}}\cdots Y_{ri_{k}}=\frac{b_{ri_{1}}\cdots b_{ri_{k}}}{s_{r}^{k}}\IE X_{ri_{1}}\cdots X_{ri_{k}}.
\]

We thus start the proof of Theorem \ref{thm:CWM} by presenting asymptotic
expressions for correlations of the type above in the CWM. These facts
are well known and presented here for the convenience of the reader.
The Curie-Weiss random variables $\left(X_{r\frac{r+1}{2}},\ldots,X_{rr}\right)$
are exchangeable. It follows that for any set of indices $\left\{ i_{1},\ldots,i_{k}\right\} $
with cardinality $k$ we have $\IE X_{ri_{1}}\cdots X_{ri_{k}}=\IE X_{r\frac{r+1}{2}}\cdots X_{r\frac{r+1}{2}+k}$,
i.e.\! the correlation $\IE X_{ri_{1}}\cdots X_{ri_{k}}$ depends
only on the number of different random variables included and not
their specific identities. We are thus free to look only at the first
$k$ random variables. The value of these correlations depends on
the regime of the model, and we have the following proposition. Below
$\left(k-1\right)!!$ stands for $\left(k-1\right)\left(k-3\right)\cdots5\cdot3\cdot1$,
and we set $\bar{\beta}:=\frac{\beta}{1-\beta}$ for all $\beta<1$. 
\begin{prop}
\label{prop:correlations}Let $k\in\IN$. If $k$ is odd, then $\IE X_{r\frac{r+1}{2}}\cdots X_{r\frac{r+1}{2}+k}$
equals $0$ for all values of $\beta\geq0$. If $k$ is even, then
we have: 
\begin{enumerate}
\item If $\beta<1$, then 
\[
\IE X_{r\frac{r+1}{2}}\cdots X_{r\frac{r+1}{2}+k}\approx\left(k-1\right)!!\,\bar{\beta}^{k/2}\left(\frac{2}{r}\right)^{k/2}.
\]
\item If $\beta=1$, then 
\[
\IE X_{r\frac{r+1}{2}}\cdots X_{r\frac{r+1}{2}+k}\approx c_{k}\left(\frac{2}{r}\right)^{k/4}.
\]
\item If $\beta>1$, then 
\[
\IE X_{r\frac{r+1}{2}}\cdots X_{r\frac{r+1}{2}+k}\xrightarrow[r\rightarrow\infty]{}x\left(\beta\right)^{k}.
\]
\end{enumerate}
The constant $c_{k}$ is given by $12^{k/4}\frac{\Gamma\left(\frac{k+1}{4}\right)}{\Gamma\left(\frac{1}{4}\right)}$,
where $\Gamma$ stands for the gamma function, and $x\left(\beta\right)$
is the positive solution of the equation $\tanh\beta x=x$. 
\end{prop}

\begin{proof}
This is Theorem 5.17 in \cite{KiMM} and its proof can be found there. 
\end{proof}
The second ingredient we need in order to calculate the moments 
\begin{align*}
M_{k} & :=\IE\left(\sum_{i=\frac{r+1}{2}}^{r}Y_{ri}\right)^{k}=\frac{1}{s_{r}^{k}}\sum_{i_{1},\ldots,i_{k}=\frac{r+1}{2}}^{r}b_{ri_{1}}\cdots b_{ri_{k}}\,\IE X_{ri_{1}}\cdots X_{ri_{k}}
\end{align*}
is a count of the number of what we will call profiles. If we have
an index vector $I=\left(i_{1},\ldots,i_{k}\right)$, where all indices
belong to $\left\{ 1,\ldots,k\right\} $ (possibly including repeated
indices), there is a profile $\text{\ensuremath{\underbar{r}}}=\left(r_{1},\ldots,r_{k}\right)\in\left\{ 0,1,\ldots,k\right\} ^{k}$
which registers the number of indices in $I$ according to their multiplicity:
$r_{j}$ stands for the number of indices in $I$ that occur exactly
$j$ times for $j\in\left\{ 1,\ldots,k\right\} $. Let $\Pi^{k}$
be the set of all profiles of length $k$. All profiles $\text{\ensuremath{\underbar{r}}}\in\Pi^{k}$
satisfy the identity $\sum_{j=1}^{k}jr_{j}=k$. Also, the number of
profiles is finite. A simple upper bound on $\left|\Pi^{k}\right|$
is given by $\left(k+1\right)^{k}$.

We will need to know how many index vectors there are for each profile.
The following lemma provides the answer: 
\begin{lem}
\label{lem:multiplicity}Let $k\in\IN$ and $\text{\ensuremath{\underbar{r}}}\in\Pi^{k}$.
The number of index vectors $\left(i_{1},\ldots,i_{k}\right)$ with
entries in $\left\{ 1,\ldots,k\right\} $ and profile $\text{\ensuremath{\underbar{r}}}$
is given by 
\[
\frac{k!}{r_{1}!\cdots r_{k}!\,1!^{r_{1}}\cdots k!^{r_{k}}}.
\]
\end{lem}

\begin{proof}
Let $k\in\IN$ and $\text{\ensuremath{\underbar{r}}}\in\Pi^{k}$.
We construct an index vector $I=\left(i_{1},\ldots,i_{k}\right)$
with entries in $\left\{ 1,\ldots,k\right\} $ and profile $\text{\ensuremath{\underbar{r}}}$
in two steps: 
\begin{enumerate}
\item We first partition the set $\left\{ 1,\ldots,k\right\} $ into $k$
sets $A_{j}$, indexed by $j=1,\ldots,k$, with cardinality $\left|A_{j}\right|=jr_{j}$
in accordance with the profile $\text{\ensuremath{\underbar{r}}}$.
Each set $A_{j}$ indicates the positions $\ell\in\left\{ 1,\ldots,k\right\} $
such that the index $i_{\ell}$ occurs exactly $j$ times in $I$.
There are 
\begin{equation}
\left(\begin{array}{cccc}
 & k!\\
r_{1}! & \left(2r_{2}\right)! & \ldots & \left(kr_{k}\right)!
\end{array}\right)\label{eq:multinom}
\end{equation}
ways to partition $\left\{ 1,\ldots,k\right\} $ as described. 
\item We join the selected the positions $\ell\in A_{j}$ such that the
$i_{\ell}$ have the same value, i.e.\! we have to form partitions
of $A_{j}$ into subsets of $j$ elements each for all $j=1,\ldots,k$.
Let the elements of $A_{j}$ be arranged in ascending order. We form
the first set of the partition of $A_{j}$, which contains the smallest
element of $A_{j}$, by selecting $j-1$ from all remaining $jr_{j}-1$
elements of $A_{j}$. There are 
\[
\frac{\left(jr_{j}-1\right)\cdots\left(j\left(r_{j}-1\right)+1\right)}{\left(j-1\right)!}
\]
ways to make this selection. Once the first $n\in\left\{ 1,\ldots,r_{j}-1\right\} $
of the subsets of $A_{j}$ have been selected in this fashion, we
generate subset $n+1$ by choosing the smallest of the elements not
yet assigned to any of the previous subsets and choosing $j-1$ from
all remaining $j\left(r_{j}-n\right)-1$ elements of $A_{j}$. There
are 
\[
\frac{\left(j\left(r_{j}-n\right)-1\right)\cdots\left(j\left(r_{j}-n-1\right)+1\right)}{\left(j-1\right)!}
\]
ways to make this selection. Thus, we have defined the algorithm to
determine the partition of $A_{j}$ into subsets of cardinality $j$.
We observe that there are 
\begin{equation}
\frac{\left(jr_{j}-1\right)\cdots\left(j\left(r_{j}-1\right)+1\right)}{\left(j-1\right)!}\cdots\frac{\left(j-1\right)\cdots1}{\left(j-1\right)!}=\frac{\left(jr_{j}\right)!}{j^{r_{j}}r_{j}!\,\left(j-1\right)!^{r_{j}}}=\frac{\left(jr_{j}\right)!}{r_{j}!\,j!^{r_{j}}}\label{eq:partition}
\end{equation}
such partitions. 
\end{enumerate}
We multiply (\ref{eq:multinom}) and (\ref{eq:partition}) for each
$j$ and obtain 
\[
\frac{k!}{r_{1}!\left(2r_{2}\right)!\cdots\left(kr_{k}\right)!}\frac{r_{1}!\left(2r_{2}\right)!\cdots\left(kr_{k}\right)!}{r_{1}!\,1!^{r_{1}}\cdots r_{k}!\,k!^{r_{k}}}=\frac{k!}{r_{1}!\cdots r_{k}!\,1!^{r_{1}}\cdots k!^{r_{k}}}.
\]
\end{proof}
We define the correlation $\IE X\left(\text{\ensuremath{\underbar{r}}}\right)$
corresponding to a profile $\text{\ensuremath{\underbar{r}}}\in\Pi^{k}$
by choosing any $\sum_{j=1}^{k}r_{j}$ different $X_{ri}$. Then raise
the first $r_{1}$ of them to the power 1, the next $r_{2}$ to the
power 2, etc. Take the product of all these powers and take their
expectation. This expectation is what we will refer to as $\IE X\left(\text{\ensuremath{\underbar{r}}}\right)$.
Due to the exchangeability of the Curie-Weiss random variables, the
identity of the selected $X_{ri}$ in this definition is inconsequential.

Using Lemma \ref{lem:multiplicity}, we can express $M_{k}$ as 
\begin{align}
M_{k} & =\frac{1}{s_{r}^{k}}\sum_{i_{1},\ldots,i_{k}=\frac{r+1}{2}}^{r}b_{ri_{1}}\cdots b_{ri_{k}}\IE X_{ri_{1}}\cdots X_{ri_{k}}\nonumber \\
 & \approx\frac{1}{s_{r}^{k}}\sum_{\text{\ensuremath{\underbar{r}}}\in\Pi^{k}}\frac{k!}{r_{1}!\cdots r_{k}!\,1!^{r_{1}}\cdots k!^{r_{k}}}\left(\sum_{i=\frac{r+1}{2}}^{r}b_{ri}\right)^{r_{1}}\cdots\left(\sum_{i=\frac{r+1}{2}}^{r}b_{ri}^{k}\right)^{r_{k}}\IE X\left(\text{\ensuremath{\underbar{r}}}\right).\label{eq:moments}
\end{align}
Each summand in (\ref{eq:moments}) is indexed by a profile $\text{\ensuremath{\underbar{r}}}\in\Pi^{k}$.
Thus there are only finitely many of them (as noted above, at most$\left(k+1\right)^{k}$,
which is independent of $r$), and each summand depends on $r$. Most
of these summands do not contribute to the moment $M_{k}$ in the
sense that the corresponding summand goes to 0 as $r\rightarrow\infty$.
Our first task is to determine which of the summands contribute and
discard the others. Secondly, we calculate the limit of the contributing
summands and show they converge to the moment of order $k$ of the
claimed limiting distribution.

Let $\beta<1$. The limit given in Theorem \ref{thm:CWM} is standard
normal, and we have to show that $M_{k}$ converges to the moment
of order $k$ of a standard normal distribution. The moments of the
standard normal are given by 
\[
\begin{cases}
\left(k-1\right)!!, & k\textup{ even},\\
0, & k\textup{ odd.}
\end{cases}
\]
First we note that due to the symmetric nature of the Curie-Weiss
distribution defined in (\ref{eq:CWM}), all odd moments $M_{k}$
are 0. Let $k$ be even. Employing Proposition \ref{prop:correlations}
and noting that for all $i\in\left\{ \frac{r+1}{2},\ldots,r\right\} $
and all $j\in\IN_{0}$ 
\[
X_{ri}^{j}=\begin{cases}
1 & \textup{if }j\textup{ is even,}\\
X_{ri} & \text{otherwise,}
\end{cases}
\]
the correlation $\IE X\left(\text{\ensuremath{\underbar{r}}}\right)$
can be expressed for each profile $\text{\ensuremath{\underbar{r}}}\in\Pi^{k}$
as 
\[
\left(o-1\right)!!\bar{\beta}^{o/2}\left(\frac{2}{r}\right)^{o/2},
\]
where $o$ stands for the sum of $r_{j}$ over all odd $j$. By (\ref{eq:moments}),
we have 
\[
M_{k}\approx\frac{1}{s_{r}^{k}}\sum_{\text{\ensuremath{\underbar{r}}}\in\Pi^{k}}\frac{k!}{r_{1}!\cdots r_{k}!\,1!^{r_{1}}\cdots k!^{r_{k}}}\left(\sum_{i=\frac{r+1}{2}}^{r}b_{ri}\right)^{r_{1}}\cdots\left(\sum_{i=\frac{r+1}{2}}^{r}b_{ri}^{k}\right)^{r_{k}}\left(o-1\right)!!\bar{\beta}^{o/2}\left(\frac{2}{r}\right)^{o/2}.
\]
We inspect the factors in each summand which depend on $r$. These
are: 
\[
\frac{1}{s_{r}^{k}}\left(\sum_{i=\frac{r+1}{2}}^{r}b_{ri}\right)^{r_{1}}\cdots\left(\sum_{i=\frac{r+1}{2}}^{r}b_{ri}^{k}\right)^{r_{k}}\left(\frac{1}{r}\right)^{o/2}=\Theta\left(\frac{1}{r^{k/4}}\right)\Theta\left(r^{\frac{1}{2}\sum_{j=1}^{k}r_{j}}\right)\frac{1}{r^{o/2}}.
\]
Above we used the definition of $s_{r}$ and the asymptotic expression
provided by Lemma \ref{lem:sum_bri^k}. The powers of $r$ are 
\begin{equation}
-\frac{k}{4}+\frac{1}{2}\sum_{j=1}^{k}r_{j}-\frac{o}{2}=-\frac{k}{4}+\frac{1}{2}\sum_{j=1}^{k/2}r_{2j}.\label{eq:powers_r}
\end{equation}
The summand converges (to a positive constant or to 0) if and only
if (\ref{eq:powers_r}) is smaller or equal to 0. Since we have 
\[
k=\sum_{j=1}^{k}jr_{j}\geq2\sum_{j=1}^{k}r_{2j},
\]
all summands in $M_{k}$ converge. Moreover, as (\ref{eq:powers_r})
being smaller or equal to 0 is equivalent to the last inequality presented,
(\ref{eq:powers_r}) equals 0 if and only if $2r_{2}=k$. Hence, the
only summand in (\ref{eq:moments}) which contributes is the one corresponding
to the profile $\left(0,k/2,0,\ldots,0\right)$. For all other profiles,
(\ref{eq:powers_r}) is negative, and thus the summand converges to
0. We have thus shown 
\begin{align*}
M_{k} & \approx\frac{1}{s_{r}^{k}}\frac{k!}{\left(k/2\right)!\,2!^{k/2}}\left(\sum_{i=\frac{r+1}{2}}^{r}b_{ri}^{2}\right)^{k/2}\\
 & =\frac{k!}{\left(k/2\right)!\,2!^{k/2}}=\left(k-1\right)!!.
\end{align*}
This proves the moments $M_{k}$ converge to the moments of the standard
normal distribution and concludes the proof of the statement for $\beta<1$
in Theorem \ref{thm:CWM}.

Now let $\beta=1$. We analyse the moments $M_{k}$ given in (\ref{eq:moments}).
Odd moments are 0, so let $k$ be even. According to Proposition \ref{prop:correlations},
the correlations $\IE X\left(\text{\ensuremath{\underbar{r}}}\right)$
can be expressed as 
\[
c_{o}\left(\frac{2}{r}\right)^{o/4},
\]
where $o$ again stands for the sum of $r_{j}$ over all odd $j$.
We inspect the factors in each summand in (\ref{eq:moments}) which
depend on $r$. These are: 
\[
\frac{1}{s_{r}^{k}}\left(\sum_{i=\frac{r+1}{2}}^{r}b_{ri}\right)^{r_{1}}\cdots\left(\sum_{i=\frac{r+1}{2}}^{r}b_{ri}^{k}\right)^{r_{k}}\left(\frac{1}{r}\right)^{o/4}=\Theta\left(\frac{1}{r^{k/4}}\right)\Theta\left(r^{\frac{1}{2}\sum_{j=1}^{k}r_{j}}\right)\frac{1}{r^{o/4}}.
\]
The powers of $r$ are 
\begin{equation}
-\frac{k}{4}+\frac{1}{2}\sum_{j=1}^{k}r_{j}-\frac{o}{4}=-\frac{k}{4}+\frac{1}{2}\sum_{j=1}^{k/2}r_{2j}+\frac{1}{4}o.\label{eq:powers_r_beta=00003D00003D1}
\end{equation}
The above expression is smaller or equal to 0 if and only if 
\[
k=\sum_{j=1}^{k}jr_{j}\geq o+2\sum_{j=1}^{k/2}r_{2j}.
\]
Since this equality holds for all profiles, none of the summands diverges.
The contributing summands are those for which equality holds in the
inequality above. The profiles for which equality holds are $\left\{ \left.\left(2j,\frac{k}{2}-j,0,\ldots,0\right)\,\right|\,j=0,\ldots,\frac{k}{2}\right\} $,
i.e.\! each index occurs either once or twice. Therefore, 
\begin{align*}
M_{k} & \approx\frac{1}{s_{r}^{k}}\sum_{j=1}^{k/2}\frac{k!}{\left(2j\right)!\left(\frac{k}{2}-j\right)!\,2!^{k/2-j}}\left(\sum_{i=\frac{r+1}{2}}^{r}b_{ri}\right)^{2j}\left(\sum_{i=\frac{r+1}{2}}^{r}b_{ri}^{2}\right)^{k/2-j}c_{2j}\left(\frac{2}{r}\right)^{j/2}.
\end{align*}
So we have convergence of the moments to fixed limits and as a consequence
also convergence in distribution. The moments given above are \emph{not}
those of a normal distribution, which can be verified by calculating
$M_{2}$ and $M_{4}$ and noting that $\lim_{r\rightarrow\infty}M_{4}\neq3\left(\lim_{r\rightarrow\infty}M_{2}\right)^{2}$,
whereas the moments of any centred normal distribution satisfy this
equality.

Finally, we treat the case $\beta>1$. Since the normalisation of
$\lambda_{r}\left(\mathbf{a}\right)$ is now $\sqrt{r}$, the triangular
array $\left(Y_{ri}\right)_{r\in\IN,i\in\left\{ \frac{r+1}{2},\ldots,r\right\} }$
is defined by $Y_{ri}:=\frac{b_{ri}\left(1-2a_{i}\right)}{\sqrt{r}}=\frac{b_{ri}X_{ri}}{\sqrt{r}}$
and the moments $M_{k}$ by 
\begin{align*}
M_{k} & :=\IE\left(\sum_{i=1}^{n}Y_{ni}\right)^{k}=\frac{1}{r^{k/2}}\sum_{i_{1},\ldots,i_{k}=\frac{r+1}{2}}^{r}b_{ri_{1}}\cdots b_{ri_{k}}\,\IE X_{ri_{1}}\cdots X_{ri_{k}}\\
 & \approx\frac{1}{r^{k/2}}\sum_{\text{\ensuremath{\underbar{r}}}\in\Pi^{k}}\frac{k!}{r_{1}!\cdots r_{k}!\,1!^{r_{1}}\cdots k!^{r_{k}}}\left(\sum_{i=\frac{r+1}{2}}^{r}b_{ri}\right)^{r_{1}}\cdots\left(\sum_{i=\frac{r+1}{2}}^{r}b_{ri}^{k}\right)^{r_{k}}\IE X\left(\text{\ensuremath{\underbar{r}}}\right).
\end{align*}
By Proposition \ref{prop:correlations}, the correlations $\IE X\left(\text{\ensuremath{\underbar{r}}}\right)$
converge to 
\[
x\left(\beta\right)^{o}
\]
in the limit $r\rightarrow\infty$, where $o$ again stands for the
sum of $r_{j}$ over all odd $j$.

The odd moments $M_{k}$ are 0. For the even moments, each summand
of $M_{k}$ has the following factors which depend on $r$: 
\[
\frac{1}{r^{k/2}}\left(\sum_{i=\frac{r+1}{2}}^{r}b_{ri}\right)^{r_{1}}\cdots\left(\sum_{i=\frac{r+1}{2}}^{r}b_{ri}^{k}\right)^{r_{k}}=\frac{1}{r^{k/2}}\Theta\left(r^{\frac{1}{2}\sum_{j=1}^{k}r_{j}}\right).
\]
Thus, the powers of $r$ in each summand are 
\[
-\frac{k}{2}+\frac{1}{2}\sum_{j=1}^{k}r_{j}.
\]
This expression is non-positive for all profiles, so none of the summands
diverges. The inequality holds with equality for the profile $\left(k,0,\ldots,0\right)$,
and it is strict for all other profiles. Hence, the moment $M_{k}$
is asymptotically equal to 
\[
M_{k}\approx\frac{1}{r^{k/2}}\left(\sum_{i=\frac{r+1}{2}}^{r}b_{ri}\right)^{k}x\left(\beta\right)^{k}\approx\frac{1}{r^{k/2}}\left(\frac{2r}{\pi}\right)^{k/2}x\left(\beta\right)^{k}=\left(\left(\frac{2}{\pi}\right)^{1/2}x\left(\beta\right)\right)^{k}=:m^{k}.
\]
Since all odd moments are 0 and all even moments are $m^{k}$, we
conclude that $\frac{\lambda_{r}\left(\mathbf{a}\right)}{\sqrt{r}}$
converges in distribution to $\frac{1}{2}\left(\delta_{-m}+\delta_{m}\right)$.

\section{Proof of Theorem \ref{thm:CWM_ext}}

Contrary to the CWM defined in (\ref{eq:CWM}), the model with an
external magnetic field defined in (\ref{eq:CWM_ext}) features flip
parameters with expectations different than $1/2$, analogous to the
case of independent flip parameters treated in Theorem \ref{thm:iid_neq_1/2}.
The CWM with an external magnetic field shows a different behaviour
and we cannot employ Proposition \ref{prop:correlations} for the
correlations.

Instead of the equation $\tanh\beta x=x$, we have to analyse the
more general equation 
\begin{equation}
\tanh\left(\beta\left(x+h\right)\right)=x.\label{eq:CW_ext}
\end{equation}
For any value of $\beta\geq0$, $h\neq0$, there is a solution $x\left(\beta,h\right)$
of (\ref{eq:CW_ext}) with the same sign as $h$. (We remark that
if $\beta>1$ and $\left|h\right|$ is small enough, then there are
one or two solutions of the opposite signs as well. These are of no
consequence for our analysis. See Sections IV.4 and V.9 of \cite{Ell1985}
for a detailed analysis of this model.) We have the following limit
result for the correlations: 
\begin{prop}
\label{prop:correlations_ext}Let $\beta\geq0$, $h\neq0$, and $x\left(\beta,h\right)$
the solution of (\ref{eq:CW_ext}) defined above. Then we have for
all $k\in\IN$ 
\[
\IE X_{r\frac{r+1}{2}}\cdots X_{r\frac{r+1}{2}+k}\xrightarrow[r\rightarrow\infty]{}x\left(\beta,h\right)^{k}.
\]
\end{prop}

As this proposition suggests, contrary to the model without an external
magnetic field treated previously, we do not have distinct behaviour
in the form of three different regimes as a function of $\beta$.

We define the triangular array $\left(Y_{ri}\right)_{r\in\IN,i\in\left\{ \frac{r+1}{2},\ldots,r\right\} }$
by $Y_{ri}:=\frac{b_{ri}\left(1-2a_{i}-x\left(\beta,h\right)\right)}{s_{r}}=\frac{b_{ri}\left(X_{ri}-x\left(\beta,h\right)\right)}{s_{r}}$
and the moments $M_{k}$ for all $k\in\IN$ by

\begin{align*}
M_{k} & :=\IE\left(\sum_{i=\frac{r+1}{2}}^{r}Y_{ri}\right)^{k}=\frac{1}{s_{r}^{k}}\sum_{j=0}^{k}\left(\begin{array}{c}
k\\
j
\end{array}\right)\,\left(-1\right)^{k-j}\,x\left(\beta,h\right)^{k-j}\,\IE\lambda_{r}\left(\mathbf{a}\right)^{j}.
\end{align*}
The last factor in each summand above, $\IE\lambda_{r}\left(\mathbf{a}\right)^{j}$,
can be expanded as we did in the proof of Theorem \ref{thm:CWM}:
\begin{align*}
\IE\lambda_{r}\left(\mathbf{a}\right)^{j} & =\sum_{i_{1},\ldots,i_{j}=\frac{r+1}{2}}^{r}b_{ri_{1}}\cdots b_{ri_{j}}\IE X_{ri_{1}}\cdots X_{ri_{j}}\\
 & \approx\sum_{\text{\ensuremath{\underbar{r}}}\in\Pi^{j}}\frac{j!}{r_{1}!\cdots r_{j}!\,1!^{r_{1}}\cdots j!^{r_{j}}}\left(\sum_{i=\frac{r+1}{2}}^{r}b_{ri}\right)^{r_{1}}\cdots\left(\sum_{i=\frac{r+1}{2}}^{r}b_{ri}^{j}\right)^{r_{j}}\IE X\left(\text{\ensuremath{\underbar{r}}}\right).
\end{align*}
As a consequence of Proposition \ref{prop:correlations_ext}, the
correlations $\IE X\left(\text{\ensuremath{\underbar{r}}}\right)$
for all $\text{\ensuremath{\underbar{r}}}\in\Pi^{j}$ converge to
\[
\IE X\left(\text{\ensuremath{\underbar{r}}}\right)\xrightarrow[r\rightarrow\infty]{}x\left(\beta,h\right)^{o},
\]
where $o$ is the sum of the $r_{\ell}$ over all odd natural numbers
$\ell$ in the range $\left[1,j\right]$.

Thus, $M_{k}$ is a double sum over $j=0,\ldots,k$ and $\text{\ensuremath{\underbar{r}}}\in\Pi^{j}$.
The remainder of the proof consists of determining the factors which
depend on $r$ in each summand, 
\[
\Theta\left(\frac{1}{s_{r}}\right)\,\Theta\left(r^{\frac{1}{2}\sum_{\ell=1}^{j}r_{\ell}}\right),
\]
analysing which of the summands contribute asymptotically to $M_{k}$,
calculating their sum, and showing that it converges to the corresponding
moment of order $k$ of a centred normal distribution. Since this
procedure is very similar to the three cases we have already treated
in such fashion in the proof of Theorem \ref{thm:CWM}, we will omit
this part and conclude here.

\bibliographystyle{plain}

\end{document}